\documentclass[11pt]{article}
\usepackage{bbm}
\usepackage{framed} 
\usepackage{dsfont}
\usepackage{url} 
\usepackage[bottom=1in, footskip=0.5in]{geometry}
\usepackage{complexity}
\usepackage{booktabs}
\usepackage{amsmath,amssymb}
\usepackage{amsthm}
\usepackage{thmtools} 
\usepackage{thm-restate}
\usepackage{nicefrac}
\usepackage{microtype}
\usepackage{calc}
\usepackage{enumerate}
\usepackage{enumitem}
\usepackage[usenames,dvipsnames]{xcolor}
\usepackage[colorlinks,citecolor=blue,linkcolor=BrickRed]{hyperref}
\usepackage{textgreek}
\usepackage{setspace}
\usepackage{parskip}

\usepackage{fullpage}
\allowdisplaybreaks

\usepackage{graphicx}
\usepackage[font=footnotesize,labelfont=bf]{subcaption}
\usepackage[font=footnotesize,labelfont=bf]{caption}
\usepackage[nobreak=true,skipabove=7pt]{mdframed}
\usepackage{appendix}
\usepackage[colorinlistoftodos]{todonotes}

\usepackage{algorithm}
\usepackage{algcompatible}
\usepackage[noend]{algpseudocode}

\usepackage{xr}
\usepackage{array}
\usepackage{xspace}
\usepackage[capitalise]{cleveref}

\usepackage{chngcntr}
\usepackage{mathtools}

\usepackage{nameref}

\usepackage{tikz}
\usetikzlibrary{patterns}

\DeclareMathOperator*{\Ber}{Ber}

\DeclareMathOperator*{\expectation}{\mathbb{E}}
\let\poly\relax
\DeclareMathOperator*{\poly}{poly}

\DeclareMathOperator*{\probability}{\Pr}

\newcommand\eps{\epsilon}

\newcommand\cA{\mathcal{A}}
\newcommand\abs[1]{\lvert #1 \rvert}

\newcommand{\prob}{\probability\probarg}
\DeclarePairedDelimiterX{\probarg}[1]{(}{)}{%
	\ifnum\currentgrouptype=16 \else\begingroup\fi
	\activatebar#1
	\ifnum\currentgrouptype=16 \else\endgroup\fi
}

\newcommand{\probover}[1]{\probability_{#1}\probarg}

\newcommand{\expect}{\expectation\expectarg}
\DeclarePairedDelimiterX{\expectarg}[1]{[}{]}{%
	\ifnum\currentgrouptype=16 \else\begingroup\fi
	\activatebar#1
	\ifnum\currentgrouptype=16 \else\endgroup\fi
}

\newcommand{\expectover}[1]{\expectation_{#1}\expectarg}

\DeclarePairedDelimiterX{\wklarg}[1]{(}{)}{%
	\ifnum\currentgrouptype=16 \else\begingroup\fi
	\activatebars#1
	\ifnum\currentgrouptype=16 \else\endgroup\fi
}

\newcommand{\innermid}{\nonscript\;\delimsize\vert\nonscript\;}
\newcommand{\activatebar}{%
	\begingroup\lccode`\~=`\|
	\lowercase{\endgroup\let~}\innermid 
	\mathcode`|=\string"8000
}

\newcommand{\innermids}{\nonscript\;\delimsize\vert\delimsize\vert\nonscript\;}
\newcommand{\activatebars}{%
	\begingroup\lccode`\~=`\|
	\lowercase{\endgroup\let~}\innermids 
	\mathcode`|=\string"8000
}

\newcommand\opt{\textsc{Opt}\xspace}
\newcommand\copt{c\left(\textsc{Opt}\right)\xspace}

\newcommand\lpopt{\textsc{LP}_{\textsc{Opt}}}
\newcommand\alg{\textsc{Alg}\xspace}
\newcommand\calg{c(\textsc{Alg})\xspace}

\counterwithin{equation}{section}

\newcommand{\Continue}{\textbf{continue}\xspace}

\usepackage{eqparbox}

\newcommand\wKL[2]{\textsc{KL}_c\left(#1 \mid \mid #2\right)}
\newcommand\supp[1]{\text{support}\left( #1 \right)}

\newcommand\cE{\mathcal{E}}

\newcommand\loc{\textsc{LearnOrCover}\xspace}

\newcommand\setcov{\textsc{SetCover}\xspace}

\newcommand\nmfl{\textsc{NonMetricFacilityLocation}\xspace}
\newcommand\ewst{\textsc{EdgeWeightedSteinerTree}\xspace}
\newcommand\ewsf{\textsc{EdgeWeightedSteinerForest}\xspace}
\newcommand\nwst{\textsc{NodeWeightedSteinerTree}\xspace}
\newcommand\nwsf{\textsc{NodeWeightedSteinerForest}\xspace}

\newcommand\pcnwsf{\textsc{PCNWSF}\xspace}

\newtheorem{theorem}{Theorem}[section]

\newtheorem{lemma}[theorem]{Lemma}

\newtheorem{claim}[theorem]{Claim}

\newtheorem{invariant}{Invariant}

\newtheorem{assumption}[theorem]{Assumption}

\newlength{\continueindent}
\usepackage{etoolbox}
\makeatletter
\newcommand*{\ALG@customparshape}{\parshape 2 \leftmargin \linewidth \dimexpr\ALG@tlm+\continueindent\relax \dimexpr\linewidth+\leftmargin-\ALG@tlm-\continueindent\relax}
\apptocmd{\ALG@beginblock}{\ALG@customparshape}{}{\errmessage{failed to patch}}
\makeatother

\makeatletter
\def\thm@space@setup{%
	\thm@preskip=\parskip \thm@postskip=0pt
}
\makeatother

\usepackage{etoolbox}
\usepackage{tikz}
\usetikzlibrary{tikzmark}
\usetikzlibrary{calc}

\errorcontextlines\maxdimen

\newcommand{\ALGtikzmarkcolor}{black}
\newcommand{\ALGtikzmarkextraindent}{4pt}
\newcommand{\ALGtikzmarkverticaloffsetstart}{-.5ex}
\newcommand{\ALGtikzmarkverticaloffsetend}{-.5ex}
\makeatletter
\newcounter{ALG@tikzmark@tempcnta}

\newcommand\ALG@tikzmark@start{%
	\global\let\ALG@tikzmark@last\ALG@tikzmark@starttext%
	\expandafter\edef\csname ALG@tikzmark@\theALG@nested\endcsname{\theALG@tikzmark@tempcnta}%
	\tikzmark{ALG@tikzmark@start@\csname ALG@tikzmark@\theALG@nested\endcsname}%
	\addtocounter{ALG@tikzmark@tempcnta}{1}%
}

\def\ALG@tikzmark@starttext{start}
\newcommand\ALG@tikzmark@end{%
	\ifx\ALG@tikzmark@last\ALG@tikzmark@starttext
	\else
	\tikzmark{ALG@tikzmark@end@\csname ALG@tikzmark@\theALG@nested\endcsname}%
	\tikz[overlay,remember picture] \draw[\ALGtikzmarkcolor] let \p{S}=($(pic cs:ALG@tikzmark@start@\csname ALG@tikzmark@\theALG@nested\endcsname)+(\ALGtikzmarkextraindent,\ALGtikzmarkverticaloffsetstart)$), \p{E}=($(pic cs:ALG@tikzmark@end@\csname ALG@tikzmark@\theALG@nested\endcsname)+(\ALGtikzmarkextraindent,\ALGtikzmarkverticaloffsetend)$) in (\x{S},\y{S})--(\x{S},\y{E});%
	\fi
	\gdef\ALG@tikzmark@last{end}%
}

\apptocmd{\ALG@beginblock}{\ALG@tikzmark@start}{}{\errmessage{failed to patch}}
\pretocmd{\ALG@endblock}{\ALG@tikzmark@end}{}{\errmessage{failed to patch}}
\makeatother

\algblock[with]{With}{EndWith}
\algblockdefx[With]{With}{EndWith}%
[1]{\textbf{with} #1 \textbf{do}}%
{}

\makeatletter
\ifthenelse{\equal{\ALG@noend}{t}}%
{\algtext*{EndWith}}
{}%
\makeatother

\title{Stochastic Gradient Meets Randomized Rounding: New Algorithms for Node-Weighted Steiner Problems}
\author{}
 \author{
        Joseph Koutsoutis\thanks{Computer Science Department, Rutgers University, New Brunswick, NJ 08816. Emails:\texttt{\{jsk245,jdl217,jy1149\}@scarletmail.rutgers.edu}, \texttt{roie.levin@rutgers.edu}. }
 	\and
        Jesse Lerner$^*$
 	\and
 	Roie Levin$^*$
        \and
        Jiawei Yu$^*$
 }
\date{}

\begin{document}
    \maketitle
    \begin{abstract}

We give a new $O(\log n)$ approximation algorithm for \textsc{NodeWeightedSteinerTree} and \textsc{NodeWeightedSteinerForest}. Our algorithm matches the bounds of Klein \& Ravi [J. Algorithms '95] which are best possible unless \P = \NP, but have the advantage that they work in the \emph{online} setting when the terminal pairs are revealed in \emph{random order}.

To obtain our results, we combine the LearnOrCover framework due to Gupta, Kehne, Levin [FOCS '21] with the Augmented Greedy algorithm of Berman \& Coulston [STOC '97] for online \emph{edge-weighted} Steiner Forest. Neither algorithm suffices on its own, but the analyses dovetail to imply our guarantee. Run offline, the algorithm reduces to a very simple randomized rounding scheme that (in spirit) reduces Node Weighted Steiner Forest to Edge Weighted Steiner Forest, and we hope this idea finds further applications.
\end{abstract}
    
    \section{Introduction}

In the Steiner Tree problem, we are given a graph $G = (V, E)$ with non-negative edge weights $\{c_e\}_e \in \mathbb{R}_{\geq0}$, as well as a set of $k$ terminals $T \subseteq V$. The goal is to buy a minimal cost subgraph that connects all terminals pairwise. In the Steiner Forest generalization (also known as generalized Steiner Tree), the input is instead a set of $k$ terminal \emph{pairs} $T \subseteq V \times V$, and this time we must buy a minimum cost subgraph that need only connect between paired terminals. 
In the \emph{online} version of these problems, the request set $T$ is revealed one-by-one. After each arrival, the algorithm must connect the incoming request by way of some path, and this decision is irrevocable. We have long known $O(1)$ approximation algorithms offline \cite{DBLP:journals/jco/KarpinskiZ97,DBLP:journals/jal/PromelS00,DBLP:journals/algorithmica/Zelikovsky93,DBLP:conf/soda/RobinsZ00,DBLP:journals/jacm/ByrkaGRS13, DBLP:journals/siamcomp/AgrawalKR95, DBLP:journals/siamcomp/GoemansW95} (and the problem is APX hard \cite{DBLP:journals/jacm/AroraLMSS98}), and that the simple greedy algorithm which connects each arriving terminal via a shortest path is $O(\log n)$ competitive \cite{DBLP:journals/siamdm/ImaseW91,DBLP:journals/tcs/AwerbuchAB04,DBLP:conf/stoc/BermanC97} (and this is information theoretically tight \cite{DBLP:journals/siamdm/ImaseW91}).

In the node-weighted variant, the input also includes vertex cost $\{c_v\}_v \in \mathbb{R}_{\geq0}$ that the algorithm must pay for each vertex in the output tree. This change seems small, but it makes the problem qualitatively harder: since this version captures \textsc{SetCover} as a special case, it has $\Omega(\log n)$ hardness of approximation \cite{feige1998thresholdforsetcover,DBLP:conf/stoc/DinurS14} and an $\Omega(\log^2 n)$  lower bound online \cite{korman2004randomizationinsetcover} (for polynomial time algorithms, assuming $\NP \not \subseteq \BPP$). The added complexity has also frustrated algorithmic progress; the first $O(\log k)$ approximation algorithm for node-weighted Steiner tree is due to Klein and Ravi \cite{DBLP:journals/jal/KleinR95} in 1995, but only very recently has the story been completed with an $O(\log k \log n)$ competitive online algorithm by Borst, Eli\'as and Venzin \cite{DBLP:conf/soda/Borst0V25} (building on the work of  \cite{DBLP:conf/focs/NaorPS11,DBLP:journals/siamcomp/HajiaghayiLP17}).

\subsection{Our Results}
Our main contributions in this paper are algorithms for offline and random-order online \nwst and \nwsf which are best possible up to constant factors.
By \emph{random order}, we mean that the input graph $G$ and the ultimate terminal set $T$ are adversarially chosen, but the terminal set is initially unknown and the order in which $T$ is revealed is uniformly random.

\begin{restatable}{theorem}{thmnwsf}
    \label{thm:main_nwsf}
    There is a polynomial-time $O(\log n)$-competitive algorithm for random order \nwsf.
\end{restatable}

This bound matches the tight offline approximation ratio up to constants in the $n = \poly(k)$ regime, and is best-possible up to $\log \log n$ terms: an information-theoretic $\Omega(\log n / \log \log n)$ lower bound follows from known constructions for set cover even when $n \gg k$ (see e.g. \cite[Theorem 5.2]{gupta2024randomordersetcover}). Run offline, our algorithm is a very simple $O(\log k)$-competitive LP rounding algorithm for \nwst, which to our knowledge is new. We also show in \cref{subsec:adv-order-online-rounding} that it implies a particularly simple $O(\log k \log n)$-competitive algorithm for the adversarial order online setting, matching the result of \cite{DBLP:conf/soda/Borst0V25}.

We then generalize the algorithm to the Prize-Collecting setting, where the algorithm is allowed to forego covering an incoming terminal at a cost.

\begin{restatable}{theorem}{thmpcnwsf}
    \label{thm:main_pcnwsf}
    There is a polynomial-time $O(\log(n) + \log^2 k)$-competitive algorithm for random order prize-collecting  \nwsf.
\end{restatable}
We suspect that this last bound can be improved, and we leave this as an open question. Our algorithms can be implemented in  $O(k \cdot [\text{shortest path computation}]) = \tilde O(km)$ time (recall $k$ is the number of terminals, $n$ the number of vertices, $m$ the number of edges), which matches the running time of the fastest known $O(\log n)$-approximation algorithm due to Bateni et al.\ \cite{DBLP:conf/icalp/BateniHL13}. To the best of our knowledge, this is the first use of \loc to design fast algorithms, which we hope inspires future work.

Besides improving existing bounds for frontier problems, we hope our work illustrates that designing algorithms for restricted models of computation (e.g. online random order) is a fruitful research program that yields fundamental and useful algorithmic primitives even for classical settings.

\subsection{Techniques and Overview}

We begin in \cref{sec:offline} with an offline $O(\log k)$-approximate LP rounding algorithm for \nwst, since it isolates some of the main ideas. The starting point is the question: why does the greedy algorithm, i.e. scan terminals in arbitrary order and connect each via a shortest path to what has been purchased so far, fail in the node-weighted version? The edge-weighted analysis (see the exposition by \cite{panigrahi2015online}) uses that the sum of radii of disjoint balls around terminals is a lower bound on $\copt$, since $\opt$ must buy a path to leave each ball; the issue is precisely that this is false in the node-weighted case, because many balls can intersect in a single vertex whose purchase greatly reduces the cost to connect all terminals. Our solution is to perform independent random rounding on the standard LP relaxation and contract sampled vertices. This effectively excises the ``shortcut'' vertices so that a modified version of the edge-weighted greedy argument goes through.

To port this idea to the random order setting, our algorithm follows the \loc paradigm of \cite{gupta2024randomordersetcover} originally designed for \setcov and subsequently generalized to \nmfl \cite{gupta2023setcoveringeyeswide}. The basic idea of \loc is to maintain a fractional solution over a set of ``actions'' (these are sets in \setcov) as follows. 
As ``requests'' (elements in \setcov) arrive online, the algorithm simultaneously samples from this fractional solution and performs a multiplicative update when the current distribution does not place enough mass on actions that are ``useful'' for the arriving request (in \setcov, the LP-mass of sets containing the incoming element are less than say $1/2$).\footnote{This can be interpreted as (stochastic) exponentiated gradient descent, see \cite[Appendix D]{gupta2024randomordersetcover}.} The analysis, which uses a KL-divergence based potential function, then shows that \emph{learning} and \emph{covering} trade off smoothly: whenever randomized rounding does not already make enough progress, the fractional solution moves closer to \opt. 

There are a few obvious challenges to implementing this strategy for \nwsf. For one, what is the set of actions from which we should sample? How should we update our distribution over actions using information from a random request? And what is the appropriate augmentation step to perform? 

In \cref{sec:st_tree}, we instantiate these design decisions first in the special case of \nwst for exposition. The requests become terminals, the actions become Steiner vertices, the augmentations become greedy shortest paths, and the crucial insight is that the distribution update should up-weight the ``shortcut'' vertices from the offline argument. Intuitively, over time this strategy should eventually learn the important part of the LP solution that we used offline. However, the new challenge is that, unlike the offline case, we do not have the luxury of assuming shortcut vertices have been fully removed ahead of time. 

The remedy is as follows. In every step, if the shortest path that greedy buys is too expensive, then there exists a shortcut node invalidating the edge-weighted analysis. If the fractional solution does not place significant mass on this shortcut node, then multiplicative update makes progress; if it does place enough mass, then sampling from the fractional solution is likely to buy shortcut nodes for future terminals. This argument is made in expectation over the random choice of every incoming terminal, and this is important.\footnote{In fact this strategy \emph{must} fail in the online node-weighted setting since it suggests an $O(\log n)$ competitive ratio, but this problem has a $\Omega(\log^2 n)$ lower bound \cite{korman2004randomizationinsetcover}.} We measure progress by combining the \loc potential function with the proxy for $\copt$ that the greedy algorithm charges its cost to.

\begin{figure}
	\centering
	
    \tikzset{every picture/.style={line width=0.75pt}} 
		
		\resizebox{0.55\linewidth}{!}{%
		\tikzset{every picture/.style={line width=0.75pt}} 

\begin{tikzpicture}[x=0.75pt,y=0.75pt,yscale=-1,xscale=1]

\draw  [fill={rgb, 255:red, 74; green, 74; blue, 74 }  ,fill opacity=0.2 ][dash pattern={on 4.5pt off 4.5pt}] (214.25,240) .. controls (214.25,177.18) and (265.18,126.25) .. (328,126.25) .. controls (390.82,126.25) and (441.75,177.18) .. (441.75,240) .. controls (441.75,302.82) and (390.82,353.75) .. (328,353.75) .. controls (265.18,353.75) and (214.25,302.82) .. (214.25,240) -- cycle ;

\draw    (241,20) -- (240,140) ;
\draw    (150,240) .. controls (107,78) and (200,50) .. (240,20) ;
\draw    (245,145) -- (150,240) ;
\draw    (245,145) -- (328,240) ;
\draw  [fill={rgb, 255:red, 255; green, 255; blue, 255 }  ,fill opacity=1 ] (181,54) .. controls (181,49.03) and (185.03,45) .. (190,45) .. controls (194.97,45) and (199,49.03) .. (199,54) .. controls (199,58.97) and (194.97,63) .. (190,63) .. controls (185.03,63) and (181,58.97) .. (181,54) -- cycle ;
\draw  [fill={rgb, 255:red, 255; green, 255; blue, 255 }  ,fill opacity=1 ] (148,91) .. controls (148,86.03) and (152.03,82) .. (157,82) .. controls (161.97,82) and (166,86.03) .. (166,91) .. controls (166,95.97) and (161.97,100) .. (157,100) .. controls (152.03,100) and (148,95.97) .. (148,91) -- cycle ;
\draw  [fill={rgb, 255:red, 255; green, 255; blue, 255 }  ,fill opacity=1 ] (131,141) .. controls (131,136.03) and (135.03,132) .. (140,132) .. controls (144.97,132) and (149,136.03) .. (149,141) .. controls (149,145.97) and (144.97,150) .. (140,150) .. controls (135.03,150) and (131,145.97) .. (131,141) -- cycle ;
\draw  [fill={rgb, 255:red, 255; green, 255; blue, 255 }  ,fill opacity=1 ] (132,191) .. controls (132,186.03) and (136.03,182) .. (141,182) .. controls (145.97,182) and (150,186.03) .. (150,191) .. controls (150,195.97) and (145.97,200) .. (141,200) .. controls (136.03,200) and (132,195.97) .. (132,191) -- cycle ;
\draw  [fill={rgb, 255:red, 0; green, 0; blue, 0 }  ,fill opacity=1 ] (230,10) -- (250,10) -- (250,30) -- (230,30) -- cycle ;
\draw  [fill={rgb, 255:red, 0; green, 0; blue, 0 }  ,fill opacity=1 ] (140,230) -- (160,230) -- (160,250) -- (140,250) -- cycle ;
\draw  [fill={rgb, 255:red, 0; green, 0; blue, 0 }  ,fill opacity=1 ] (338,230) -- (318,230) -- (318,250) -- (338,250) -- cycle ;
\draw  [fill={rgb, 255:red, 0; green, 0; blue, 0 }  ,fill opacity=1 ] (190,230) -- (210,230) -- (210,250) -- (190,250) -- cycle ;
\draw    (245,145) -- (200,240) ;
\draw  [fill={rgb, 255:red, 255; green, 255; blue, 255 }  ,fill opacity=1 ] (215,145) .. controls (215,131.19) and (226.19,120) .. (240,120) .. controls (253.81,120) and (265,131.19) .. (265,145) .. controls (265,158.81) and (253.81,170) .. (240,170) .. controls (226.19,170) and (215,158.81) .. (215,145) -- cycle ;
\draw    (331,240) .. controls (374,78) and (281,50) .. (241,20) ;
\draw  [fill={rgb, 255:red, 255; green, 255; blue, 255 }  ,fill opacity=1 ] (300,54) .. controls (300,49.03) and (295.97,45) .. (291,45) .. controls (286.03,45) and (282,49.03) .. (282,54) .. controls (282,58.97) and (286.03,63) .. (291,63) .. controls (295.97,63) and (300,58.97) .. (300,54) -- cycle ;
\draw  [fill={rgb, 255:red, 255; green, 255; blue, 255 }  ,fill opacity=1 ] (333,91) .. controls (333,86.03) and (328.97,82) .. (324,82) .. controls (319.03,82) and (315,86.03) .. (315,91) .. controls (315,95.97) and (319.03,100) .. (324,100) .. controls (328.97,100) and (333,95.97) .. (333,91) -- cycle ;
\draw  [fill={rgb, 255:red, 255; green, 255; blue, 255 }  ,fill opacity=1 ] (350,141) .. controls (350,136.03) and (345.97,132) .. (341,132) .. controls (336.03,132) and (332,136.03) .. (332,141) .. controls (332,145.97) and (336.03,150) .. (341,150) .. controls (345.97,150) and (350,145.97) .. (350,141) -- cycle ;
\draw  [fill={rgb, 255:red, 255; green, 255; blue, 255 }  ,fill opacity=1 ] (349,191) .. controls (349,186.03) and (344.97,182) .. (340,182) .. controls (335.03,182) and (331,186.03) .. (331,191) .. controls (331,195.97) and (335.03,200) .. (340,200) .. controls (344.97,200) and (349,195.97) .. (349,191) -- cycle ;
\draw    (200,240) .. controls (161,156) and (183,83) .. (240,20) ;
\draw  [fill={rgb, 255:red, 255; green, 255; blue, 255 }  ,fill opacity=1 ] (205,53) .. controls (205,48.03) and (209.03,44) .. (214,44) .. controls (218.97,44) and (223,48.03) .. (223,53) .. controls (223,57.97) and (218.97,62) .. (214,62) .. controls (209.03,62) and (205,57.97) .. (205,53) -- cycle ;
\draw  [fill={rgb, 255:red, 255; green, 255; blue, 255 }  ,fill opacity=1 ] (186,91) .. controls (186,86.03) and (190.03,82) .. (195,82) .. controls (199.97,82) and (204,86.03) .. (204,91) .. controls (204,95.97) and (199.97,100) .. (195,100) .. controls (190.03,100) and (186,95.97) .. (186,91) -- cycle ;
\draw  [fill={rgb, 255:red, 255; green, 255; blue, 255 }  ,fill opacity=1 ] (173,141) .. controls (173,136.03) and (177.03,132) .. (182,132) .. controls (186.97,132) and (191,136.03) .. (191,141) .. controls (191,145.97) and (186.97,150) .. (182,150) .. controls (177.03,150) and (173,145.97) .. (173,141) -- cycle ;
\draw  [fill={rgb, 255:red, 255; green, 255; blue, 255 }  ,fill opacity=1 ] (174,191) .. controls (174,186.03) and (178.03,182) .. (183,182) .. controls (187.97,182) and (192,186.03) .. (192,191) .. controls (192,195.97) and (187.97,200) .. (183,200) .. controls (178.03,200) and (174,195.97) .. (174,191) -- cycle ;
\draw   (108.5,46) .. controls (103.83,46) and (101.5,48.33) .. (101.5,53) -- (101.5,115.75) .. controls (101.5,122.42) and (99.17,125.75) .. (94.5,125.75) .. controls (99.17,125.75) and (101.5,129.08) .. (101.5,135.75)(101.5,132.75) -- (101.5,198.5) .. controls (101.5,203.17) and (103.83,205.5) .. (108.5,205.5) ;
\draw   (137.75,267.75) .. controls (137.75,272.42) and (140.08,274.75) .. (144.75,274.75) -- (228.88,274.75) .. controls (235.55,274.75) and (238.88,277.08) .. (238.88,281.75) .. controls (238.88,277.08) and (242.21,274.75) .. (248.88,274.75)(245.88,274.75) -- (333,274.75) .. controls (337.67,274.75) and (340,272.42) .. (340,267.75) ;

\draw (240,240) node [anchor=north west][inner sep=0.75pt]  [font=\Huge]  {$\dotsc $};
\draw (185,49) node [anchor=north west][inner sep=0.75pt]    {$\epsilon $};
\draw (223,138) node [anchor=north west][inner sep=0.75pt]    {$1+\epsilon $};
\draw (200,292.4) node [anchor=north west][inner sep=0.75pt]    {$\times k$};
\draw (57,117) node [anchor=north west][inner sep=0.75pt]    {$\times \nicefrac{1}{\epsilon }$};

\end{tikzpicture}
		}
	
	\caption{Illustration of the failure case of the greedy algorithm for \nwst. In the instance above, the black squares correspond to terminals, the small white vertices cost $\eps$ and the large white vertex costs $1+\eps$. Greedy will buy $k$ paths of length $1$, while \opt buys the broom of cost $1+\epsilon$. The analysis fails because the boundaries of balls of radius $\Omega(1)$ centered at the terminals intersect in the center node of cost $1+\epsilon$, and hence $\opt$ can pay $1+\epsilon$ to simultaneously leave all balls.}
	\label{fig:broom-failure}
\end{figure}
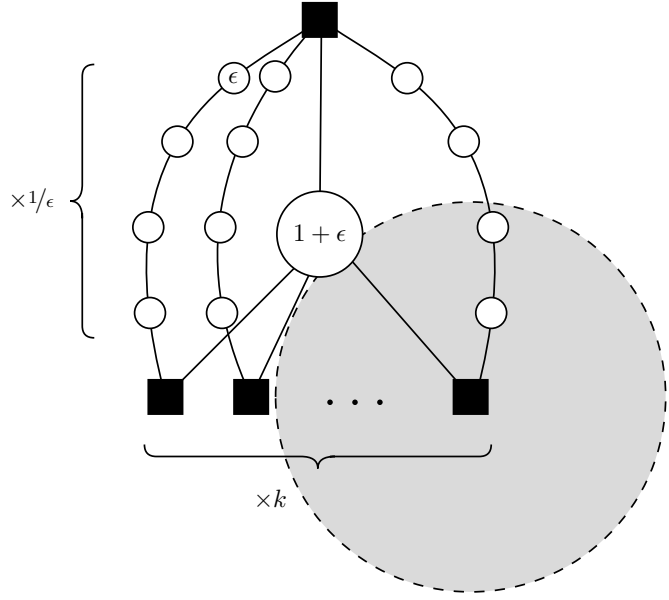

In \cref{sec:st_forest}, we turn to the more involved case of \nwsf. When a terminal pair $(u,v)$ arrives, instead of the naive greedy algorithm, we perform an \emph{augmented} greedy step as in the algorithm of Berman and Coulston \cite{DBLP:conf/stoc/BermanC97}:  (1) connect $u$ and $v$ via the shortest path of length $d$, and (2)  if both terminals $u$ and $v$ are respectively within distance $O(d)$ of previous terminals $u'$ and $v'$ (not necessarily paired) who each required connection cost $O(d)$ on arrival, then connect $u$ to $u'$ and $v$ to $v'$. This change, plus some technical complications in the analysis stemming from the fact that the cost of (2) might be at a different scale as the cost of (1), suffice to generalize the analysis for the \nwst case. Finally, in \cref{sec:pcst_forest} we generalize the analysis to the prize-collecting setting, at the expense of extra loss in our competitive ratio.

\subsection{Related Work}

Klein and Ravi were the first to design an algorithm for \nwsf, and they obtained a best possible approximation ratio of \(O(\log k)\) \cite{DBLP:journals/jal/KleinR95}. They show that the optimal Steiner tree can be decomposed into spiders (i.e. trees with at most one vertex of degree more than two) which implies that the greedy algorithm, as in textbook analyses of \setcov, has low cost in every step. The constant in the approximation ratio was later improved by Guha and Khuller \cite{DBLP:journals/iandc/GuhaK99}. Another elegant moat-growing algorithm, similar to the $2$-approximation for edge-weighted Steiner Forest \cite{DBLP:journals/siamcomp/AgrawalKR95,DBLP:journals/siamcomp/GoemansW95}, was given by Bateni et al.\ \cite{DBLP:conf/icalp/BateniHL13}.

Naor, Panigrahi, and Singh were the first to design online algorithms for \nwst, achieving a competitive ratio of $O(\log n \log^2 k)$ \cite{DBLP:conf/focs/NaorPS11}. They use spider decompositions to reduce the problem to an instance of online \nmfl at a cost of \(\log k\) in the competitive ratio, then use the online algorithm for \nmfl due to Alon et al. \cite{DBLP:journals/talg/AlonAABN06} which incurs  competitive ratio $O(\log n \log k)$.

The technique of reducing to \nmfl was refined by \cite{DBLP:journals/siamcomp/HajiaghayiLP17} to give an $O(\log n \log^2 k)$-competitive online algorithm for \nwsf, and again by
 \cite{DBLP:conf/soda/Borst0V25} to finally achieve the tight $O(\log n \log k)$
  competitive ratio for \nwsf and its prize-collecting extension. The key contribution of the last algorithm was showing how to solve \nwsf via a reduction to a single auxiliary instance of \nmfl, whereas
  \cite{DBLP:journals/siamcomp/HajiaghayiLP17} needed $\log k$ many instances.

  Finally, we mention that there has been previous interest in direct rounding schemes for \nwst, especially in the online setting. Naor et al.\ \cite{DBLP:conf/focs/NaorPS11} write in passing that the classic LP for \nwst ``appears to be too weak to allow for this kind of (simple) rounding without losing a polynomial factor in the competitive ratio''. One can view our work as a refutation of this comment.

    \section{Preliminaries}
\label{sec:prelims}

\paragraph{Assumptions on $G$.}

The reader may assume throughout the paper that all edge weights are zero. This is without loss of generality since we can subdivide every edge with non-zero weight and move its weight to the new node. We will also assume that all terminals cost zero as we can create a copy of every terminal with zero weight and connect it to its original version with a zero weight edge (and we remove the original version from the terminal set). The approximation ratio is asymptotically unchanged since $n$ only increases polynomially.

\paragraph{Notions from Graph Theory.}

For node-weighted graph $G$ and $u,v \in V(G)$, define $d_G(u, v)$ to be the cost of the cheapest $u$-$v$ path \emph{not} including the costs of $u,v$.
We define the ball of radius $r$ around vertex $u$ to be $B_G(u, r) := \{v \in V(G) : d(u, v) \leq r\}$. Define the boundary of $B_G(u, r)$ to be the set $\partial B_G(u, r) := \{v \in V(G) : d(u, v) \leq r, \ d(u, v) + c_v \geq r\}$, and the interior to be the set $\mathring B_G(u,r) : = B_G(u,r) \setminus \partial B_G(u,r)$.
For $F \subseteq V$, we define $G/F$ to be a copy of the graph $G$ with weights belonging to nodes in $F$ set to 0.

\paragraph{LP Relaxation.}

In the \nwsf problem, the terminal set $T$ contains terminal pairs of the form $(s_0, s_1)$. The \nwst problem is a special case of \nwsf, where $s_0 = r$ for all pairs in $T$. The following is an LP relaxation for both problems. Let $\mathcal{P}_{s}$ denotes all $s_0$-$s_1$, we have:
\begin{alignat}{3}
\min\quad & \sum_{v \in V} c_v x_v \nonumber \\
\text{s.t.}\quad & \sum_{P \in \mathcal{P}_{s}} f_P \geq 1, \quad && \forall s \in T \label{eq:path-demand} \\
& \sum_{\substack{P \in \mathcal{P}_{s} \\ P \ni v}} f_P \leq x_v, \quad && \forall s \in T, \, \forall v \in V \label{eq:vertex-capacity} \\
& x_v \geq 0, \quad && \forall v \in V \nonumber \\
& f_P \geq 0, \quad && \forall P \in \bigcup_s \mathcal{P}_s \nonumber
\end{alignat}
We note that this can be solved in polynomial time, since given $x$ variables, the $f$ variables are determined via max flow computations. We also note that previous works on \nwsf (e.g. \cite{DBLP:conf/icalp/BateniHL13}) prefer the equivalent cut based formulation, but we use this one for convenience in our proofs. We will need a few general facts about this LP. 

Throughout the paper, we make a few convenient assumptions. First, by a guess-and-double approach, we assume that we know a bound $\beta$ such that $\lpopt \leq \beta \leq 2 \cdot \lpopt$ (see \cref{sec:guess-and-double}). Let $x^*$ denote an arbitrary fixed optimal fractional solution. Then:
\begin{restatable}{fact}{support}
    \label{fact:support}
    Every fractional optimal solution $x^*$ of \nwst is only supported on nodes $v$ such that $c_v \leq \lpopt \leq \beta$.
\end{restatable}
This implies that we can safely ignore nodes of cost more than $\beta$, and we leave the proof of this fact to \cref{sec:deferred}. Second, we assume that all nodes cost at least $\beta / n$, since we can preemptively purchase all vertices that cost less for a total cost of at most $\beta$. In fact, this holds for any paths between terminals.
We summarize these assumptions as follows.
\begin{assumption}
    \label{assmpt:beta-n-bound}
    For every node $v \in V$, we have $c_v \in [\beta/n, \beta]$. 
    Moreover, if $P \in \mathcal{P}_s$ is a shortest path for some not connected terminal $s$, then $c(P) \in [\beta/n, \beta]$, where $c(P) = \sum_{v \in P} c_v$.  
\end{assumption}
\begin{proof}
    By preemptively purchasing all nodes with costs lower than $\beta/n$, and \cref{fact:support}, we know that $c_v \in [\beta/n, \beta]$. It remains to prove the same bound for all shortest paths between terminals.

    Fix a terminal pair $s = (s_0, s_1)$ that is not connected, and let $P \in \mathcal{P}_s$ be a shortest path connecting $s$. The lower bound is immediate. Let $x^*$ be an optimal solution. Then
    \begin{align*}
        \beta \geq \sum_{v}c_v x^*_v &\geq \sum_vc_v \cdot \left(\sum_{\substack{P' \in \mathcal{P}_s: \\
        v \in P'}}f_{P'}\right) =\sum_{P' \in \mathcal{P}_s}f_{P'}\cdot \left(\sum_{v \in P'}c_v\right)\geq c(P) \cdot \sum_{P' \in \mathcal{P}_s}f_{P'} \geq c(P),
    \end{align*}
    where the first inequality holds by the definition of $\beta$, the second by \eqref{eq:vertex-capacity}, the third from the fact that $P'$ is a shortest path in $\mathcal{P}_s$, and the last by \eqref{eq:path-demand}.
\end{proof}

\paragraph{KL Divergence.}

When not specified, logarithms in this paper are taken to be base $e$. In the following definitions, let $x,y \in \mathbb{R}^n_+$ be vectors. 
We use a weighted generalization of KL divergence. 
Given a nonnegative weight function $c$, define
\[
\wKL{x}{y} : = \sum_{i=1}^n c_i \left[x_i \log \left(\frac{x_i}{y_i}\right) -x_i + y_i\right].
\]
This quantity is always nonnegative, because for each $i$, the term $[x_i \ln \left(\nicefrac{x_i}{y_i}\right) -x_i + y_i]$ is nonnegative. To see this,
one can check that the function $f(z) = z \ln z - z + 1$ is convex, minimized at $z=1$, and nonnegative. Thus, $x_i \ln \left(\nicefrac{x_i}{y_i}\right) -x_i + y_i = y_i \cdot f\left(\nicefrac{x_i}{y_i}\right) \geq 0$.

    \section{Warmup: An Offline LP-Rounding Scheme}\label{sec:offline}

We begin our exposition with a simple rounding algorithm for \nwst to illustrate some of our main ideas. As hinted in the introduction, the thrust is to use randomized rounding to effectively remove ``shortcut'' nodes that cause problems for the analysis of the greedy algorithm for \ewst, and then to run said greedy algorithm. Before we give a precise definition for shortcut nodes, we present our offline LP rounding algorithm.

\begin{algorithm}[H]
\caption{\textsc{\nwst LP Rounding}}
\label{alg:lp_rounding}
\begin{algorithmic}[1]
    \State Input: LP solution $x$.
    \State Initialize $F' \leftarrow \emptyset$.
    \For{vertex $v \in V$}
    \State With probability $\min(1, x_v \cdot 4 \ln k)$, add $F' \leftarrow F' \cup \{v\}$.
    \EndFor
    \State Initialize $F \leftarrow \emptyset$.
    \For{\(i \in [k]\)}
    \State Set \(P_i\) to be the shortest path in \(G / F'\) connecting \(s\) to \(r\).
    \State Set $F \leftarrow F \cup P_i$.
    \EndFor
    \State \Return $F \cup F'$
\end{algorithmic}
\end{algorithm}

Our goal is to show the following theorem.
\begin{theorem}\label{thm:offline}
    \cref{alg:lp_rounding} produces a feasible solution of expected cost $O(\log k) \cdot c^\intercal x$
\end{theorem}

We are ready to define the key definition we need for our analysis. Let the terminals be labeled \(s^1,s^2,\ldots,s^k\) according to the order considered by \cref{alg:lp_rounding}, and we label $s^1 = r$ the root for convenience. Let \(d(s^j) = \min_{i<j} d_{G}(s^i,s^j)\) be the cost to connect $s^j$. We may assume without loss of generality that \(2 \cdot c^\intercal x/k \leq d(s^j) \leq c^\intercal x\).
Finally, for \(j \in [k]\) and \(\ell \in \{1, 2, \ldots, \lfloor \log_2 d(s^j) \rfloor\}\), define the \emph{$\ell$-shortcut nodes for $s^j$} to be the set
\begin{equation*}
  \Gamma(s^j,\ell) = \bigcup_{i < j} B(s^i,2^{\ell-2}) \cap B(s^j,2^{\ell-2}) \text{.} 
\end{equation*}
Observe that $\Gamma(s^j, \ell)$ consists of nodes $v$ such that if the cost to connect $s^j$ to some $i < j$ is $2^\ell$, then buying $v$ reduces this connection cost by a factor of $2$. We start by showing that randomized rounding hits these shortcut sets whenever they contain sufficient LP mass. 

\begin{claim}
\label{claim:no_gamma_mass}
    With probability $1-1/k$, for every $i \in [k]$, we have that $x(\Gamma(s^i,\lfloor \log_2 c(P_i) \rfloor)) < 1/2$. 
\end{claim}

\begin{proof}
Fix \(i \in [k]\).
Let \(\ell^* \geq 1\) be the minimum \(\ell\) such that $\sum_{v \in \Gamma(s^i,\ell)} x_v \geq 1/2$. If such an $\ell$ does not exist then the desired statement holds immediately. 
Then
\begin{align*}
  \Pr[\Gamma(s^i, \ell^*) \cap F' \neq \emptyset] = 1 - \prod_{v \in \Gamma(s^i, \ell^*)} (1 - \min(1, x_v \cdot 4 \ln k)) \geq 1 - \exp(-2 \ln k) = 1 - 1/k^2.
\end{align*}

Observe that if the sampled set $F'$ contains a node $v$ in  \(\Gamma(s^i, \ell^*)\), then by definition for some $i < j$ we have $d(s^i, s^j) \leq d(s^i, v) + d(s^j, v) \leq 2^{\ell^* - 1}$, which means \(c(P_i) \leq 2^{\ell^*-1}\). 
By our choice of \(\ell^*\) and the union bound over all $k$ choices of $s^i$, the claim holds.
\end{proof}

To finish the proof of the theorem, we argue why low LP-mass on all shortcut sets implies that the greedy algorithm's cost is controlled.

\begin{proof}[Proof of \cref{thm:offline}]
Let $\cE$ denote the event that there exists a terminal $s^i$ such that $x(\Gamma(s^i,\ell)) \geq 1/2$.
By claim~\ref{claim:no_gamma_mass}, $\cE$ holds with probability at most \(1/k\).
Conditioned on $\cE$, the algorithm pays at most \(k \cdot c^\intercal x\) for the shortest paths of all terminals, since $d(s^j) \leq c^\intercal x$ for each terminal $s^j$.

Now we condition on $\neg \cE$, and consider an arbitrary terminal \(s^i\) with $x(\Gamma(s^i, \ell_i)) < 1/2$.
For such terminals, we call \(\ell_i \coloneqq \lfloor \log_2(c(P_i)) \rfloor\) the \emph{level} of terminal \(s^i\).
Let 
\[
  D(s^i) = B(s^i, 2^{\ell_i - 2}) \setminus \Gamma(s^i,\ell_i).\]
By the definition of $\Gamma(s^i,\ell_i)$, the sets $D(s^i)$ and $D(s^j)$ are disjoint for any $i \neq j$ such that $\ell_i = \ell_j$, and hence \[\sum_{i=1}^k \sum_{v \in D(s^i)} c_v x_v  = \sum_{\ell =1}^{\lfloor \log_2 c^\intercal x\rfloor} \sum_{\substack{i:\\ \ell_i = \ell}} \sum_{\substack{v \in D(s^i)}} c_v x_v \leq \sum_{\ell =1}^{\lfloor \log_2 c^\intercal x\rfloor} c^{\intercal}x \leq \log_2(k) \cdot c^{\intercal} x,\]
in other words, the LP cost paid within the $\{D(s^i)\}_i$ regions is a good proxy for the total LP cost.
To complete our analysis, we show that we can \emph{charge} the cost of the path $P_i$ to the set $D(s^i)$, so it remains to relate $c(P_i)$ to $\sum_{v \in D(s^i)} c_v x_v$. Define
\begin{equation*}
  \mathcal{P}'_{s^i} = \{P \in \mathcal{P}_{s^i} : P \cap \Gamma(s^i,\ell_i) = \emptyset\} 
\end{equation*}
to be the set of paths connecting terminal $s^i$ to the root $r$ avoiding Steiner nodes in $\Gamma(s^i,\ell_i)$. 
Since $x$ is a feasible solution of the LP, each vertex $v$ carries at most $x_v$ amount of flow. 
By assumption $x(\Gamma(s^i, \ell_i)) < 1/2$, and therefore at least half of the \(s^i \rightarrow r\) flow is distributed on paths in $\mathcal{P}'_{(r, s^i)}$. 
Formally, 
\[\sum_{P \in \mathcal{P}'_{s^i}} f_P = \sum_{P \in \mathcal{P}_{s^i}} f_P - \sum_{P \not \in \mathcal{P}'_{s^i}} f_P \geq 1/2,\]
and hence
\begin{equation*}
  \sum_{v \in D(s^i)} c_v \cdot x_v \geq
  \sum_{P \in \mathcal{P}'_{(r,s^i)}} f_P \cdot c(P \cap D(s^i)) \geq
  \sum_{P \in \mathcal{P}'_{(r, s^i)}} f_P \cdot 2^{\ell_i-2} \geq
  c(P_i)/16 \text{.}
\end{equation*}
Putting everything together:
\begin{align*}
    \expect{\calg} &= \expect{c(F')} + \expect{c(F)}  = 4 c^{\intercal} x + 
    \expect{c(F) | \cE} \cdot \prob{\cE} + \expect{c(F) | \neg \cE} \cdot \prob{\neg \cE} \\
    &\leq 5 c^{\intercal} x +\expect*{\sum_{i = 1}^k c(P_i) | \neg \cE} \leq 16 \cdot \sum_{i = 1}^k \sum_{v \in D(s^i)} c_v \cdot x_v \leq 16 \cdot \log_2(k) \cdot c^{\intercal} x.\qedhere
\end{align*}
\end{proof}
    \section{Random Order Node-Weighted Steiner Tree}

\label{sec:st_tree}

With intuition from the offline rounding algorithm in hand, we proceed to the online random order setting. We begin with \nwst  before moving to the more general \nwsf case, since it already introduces a number of new ideas. Much of the difficulty lies in the fact that we can no longer cleanly separate rounding and greedy phases; nevertheless, we now show how to generalize the algorithm.

\subsection{The Algorithm}

Before we present the formal details, here are the main pieces and definitions.
\begin{enumerate}
    \item We maintain a fractional vector $x^t$ which is a (not necessarily feasible) fractional solution to the LP relaxation, and an integral solution $F^t$ which is the set of nodes we have bought up to round $t$.
    \item Every round $t$ in which an unconnected terminal $s$ arrives, we 
    \begin{enumerate}
        \item buy a shortest path $P$ connecting $s$ to $r$ of length $d^t(s) := d_{G/F^t}(r,s)$;
        \item sample every vertex $v$ with probability proportional to $x^t_v$;
        \item increase the value of $x^t_v$ for vertices $v$ in the set 
        $$\Gamma^t(s) \coloneqq \{v \in V : 
            d_{G / (F^{t-1} \cup \{v\})}(r, s) \leq d^t(s) / 2
        \}. $$
         The set $\Gamma^t(s)$ consists of vertices whose purchase would reduce $d^t(s)$ by a factor of 2.\footnote{The new definition $\Gamma^t(s)$ resembles that of $\ell$-shortcut nodes from \cref{sec:offline}, but is not identical since we can no longer cleanly separate the algorithm into rounding and greedy phases.}
    \end{enumerate}
\end{enumerate}
The formal version appears as \cref{alg:nwst}. We assume without loss of generality that $d^t(s) > 0$. Also, define $\Lambda^t$ to be the event that $X^t(s) \coloneqq \sum_{v \in \Gamma^t(s)} x_v^{t-1} < 1$. Thus $\neg \Lambda^t$ is the case when our algorithm samples the distribution but does not modify it.

\begin{algorithm}[ht]
\caption{\textsc{LearnOrCoverNWST}}
\label{alg:nwst}
\begin{algorithmic}[1]
    \State Receive root $r$.
    \State Initialize $F^0 = \{v \mid c_v < \beta/n\} \cup r$.
    \State Set $V \leftarrow \{v \mid \beta/n \leq c_v \leq \beta\}$, and $n' \leftarrow \abs{V}$.
    \State Set $x_v^0 \leftarrow \frac{\beta}{c_v n'} $ for $v \in V$; 0 otherwise.
    \For{each terminal $s$ in random order}
        \State Set $F^t \leftarrow F^{t-1}$, $R^t \leftarrow \emptyset$.
        \State Set $d^t(s) \leftarrow d_{G / F^{t-1}}(r, s)$.
        \State Set $P$ to be the shortest path in $G / F^{t-1}$ connecting $s$ to $r$.
        \State Add $F^t \leftarrow F^t \cup P$.
        \For{$v \in V$} 
            \State Add $R^t \leftarrow R^t \cup \{v\}$ w.p. $\min \left( d^t(s) \cdot x_v^{t-1} / \beta, 1  \right)$.
            \State Add $F^t \leftarrow F^t \cup R^t$.
        \EndFor
        \If {$\sum_{v \in \Gamma^t(s)} x_v^{t-1} \geq 1$} 
            \State $x^t \leftarrow x^{t-1}$, \Continue \Comment{$\Lambda^t$ does not hold}
        \EndIf
        \For{$v \in \Gamma^t(s)$} 
            \State $x_v^t \leftarrow x_v^{t-1} \cdot \exp\{d^t(s)/c_v\}$.
        \EndFor
        \For{$v \not \in \Gamma^t(s)$} 
            \State $x_v^t \leftarrow x_v^{t-1}$.
        \EndFor
        \State Set $Z^t \leftarrow \sum_{v \in V} c_v \cdot x_v^t / \beta$ and $x^t \leftarrow x^t / Z^t$ .\label{eq:normalize tree}
    \EndFor
    \State \Return $F^{k}$.
\end{algorithmic}
\end{algorithm}

We observe immediately that \cref{eq:normalize tree} maintains the following invariant: 
\begin{invariant}\label{invariant:inner product beta}
For all time steps $t$, it holds that $\sum_{v \in V} c_v \cdot x_v^t = \beta$.
\end{invariant}

\subsection{The Analysis}

We will use a potential function to bound our accumulated cost, but first we need a few more definitions. Upon arrival at time $t$, we assign incoming terminal $s^t$ the level $\ell^t_s \coloneqq \lfloor \log_2 (d^t(s))\rfloor$. 
Note that by \cref{assmpt:beta-n-bound}, $d^t(s) \in [\beta/n, \beta]$. Hence, there are at most $\log_2(n) + 1$ possible levels, ranging from $\lfloor\log_2(\nicefrac{\beta}{n})\rfloor$ to $\lfloor\log_2(\beta)\rfloor$.
Next, let $B^t \coloneqq B_{G / F^{t-1}}\left(s, 2^{\ell^t_s-2}\right)$ be the set of vertices within distance $2^{\ell^t_s - 2}$ to $s^t$. Finally, we say $s^t$ \emph{claims} the set
\begin{align}
    \label{eq:balldef tree}
    D^t(s^t) &\coloneqq B_t \setminus D_{< t}, 
    \intertext{where}
    D_{< t} &\coloneqq \bigcup_{\substack{t' < t: \\ \ell^{t'} = \ell^t_s}} D^{t'}. \nonumber
\end{align}
 Said different, every incoming terminal claims all vertices within radius $2^{\ell_s^t-2}$, except those that have already been claimed by vertices at the same level. By definition, terminals on the same level claim non-intersecting sets. By \cref{assmpt:beta-n-bound} this means
 \begin{align}
     \sum_{t=1}^T \sum_{v \in D^t(s^t)} c_v x^*_v   \leq \log_2(n) \cdot \lpopt, \label{eq:disjoint_charging}
 \end{align}
 so, analogously to \cref{sec:offline}, the LP cost within these sets is a good proxy for $\lpopt$. We would ideally like to run the greedy algorithm for the edge weighted case, and charge the connection cost of each $s^t$ to the LP cost on the vertices it claims, $D^t(s^t)$. Unfortunately, unlike \cref{sec:offline}, we cannot assume shortcut nodes have already been sampled away ahead of time. Instead, we will design other terms in the potential that give large drops in time steps where $D^t(s^t)$ is not enough to pay for the connection of $s^t$. 

Our potential function is comprised of two pieces. The \emph{learning potential} is:
\begin{equation*}
    \Phi_{L}(t) = \wKL{x^{*}}{x^t} + 8 \cdot \sum_{\ell=\lfloor\log_2(\nicefrac{\beta}{n})\rfloor}^{\lfloor\log_2(\beta)\rfloor} \left( \lpopt - \sum_{\substack{i \in [t] : \\ \ \ell^i = \ell}} \ \sum_{v \in D^i(s^i)} c_v \cdot x_v^* \right) \text{.}
\end{equation*}
The \emph{covering potential:}
\begin{equation*}
    \Phi_{C}(t) = \beta \cdot \log \left( \frac{\rho^t}{\beta} + \frac{1}{n} \right),
\end{equation*}
where $\rho^t = \sum_{s \in U^t} d^{t+1}(s)$ is the total cost to connect all future terminals via shortest paths assuming we have bought the current solution $F^t$, and $U^t$ is the set of future terminals arriving after time $t$.
Our final combined potential is:
\begin{equation}\label{eq:potentialdef tree}
    \Phi(t) = C_1 \cdot \Phi_{L}(t) + C_2 \cdot \Phi_{C}(t),
\end{equation}
where $C_1$ and $C_2$ are constants to be specified later. We are finally ready to begin our analysis.

\begin{restatable}{lemma}{potentialbounds}
    \label{lemma:potential_bounds}
    The initial potential is bounded as $\Phi(0) = O(\beta \cdot \log n)$, and $\Phi(t) \geq -O( \beta \cdot \log n)$ for all $t$.
\end{restatable}
It remains to show that the potential function $\Phi$ decreases in every round sufficiently to cover the spending of \alg, and we do so by bounding the changes in learning and covering potential separately.

 The main idea needed to bound the drop in the learning potential is that $D^t(s^t)$ does not contain enough LP cost to pay for the algorithm's action precisely in steps $t$ for which the optimal solution $x^*$ places large mass on nodes that lie in $\Gamma^t(s^t)$. The formal statement of this sentence is in \cref{lemma:alternate}, and in this case the algorithm makes progress by increasing LP mass on those nodes. Making precise this intuition, we show the bound:
\begin{restatable}{lemma}{changelearn}
    \label{lemma:change in learn}
    For all rounds \(t\), the expected change in $\Phi_L$ is 
    \begin{align}
        &\expectover{s^t, R^t}*{\Phi_L(t) - \Phi_L(t-1) \,\middle|\, x^{t-1}, F^{t-1}, U^{t-1}} \nonumber \\
        &\leq \expectover{s \sim U^{t-1}}*{\frac{e^{4} - 1}{4} \cdot d^t(s) \cdot \min \left( X^t(s), 1 \right) - \frac{d^t(s)}{8} \,\middle|\, x^{t-1}, F^{t-1}} \text{.} \label{eq:change in learn not preconnected}
    \end{align}
\end{restatable} 

However, this potential drop can be insufficient, and this happens when the expectation of the positive term is $\Omega(\expectover{s}{d^t(s)})$. Thankfully, the covering potential decreases commensurately in this case, and this is formalized in the next lemma.
\begin{restatable}{lemma}{changecover}
    \label{lemma:change in cover}
    For all rounds $t$, the expected change in $\Phi_C$ is 
    \begin{equation}\label{eq:change in cover not preconnected}
        \expectover{s^t, R^t}*{\Phi_C(t) - \Phi_C(t-1) \,\middle|\, x^{t-1}, F^{t-1}, U^{t-1}} \leq - \frac{1 - e^{-1}}{912} \cdot \expectover{s \sim U^{t-1}}*{d^t(s) \cdot \min(X^t(s), 1)  \,\middle|\, x^{t-1}, F^{t-1}} \text{.}
    \end{equation}
\end{restatable}

The proofs for \cref{lemma:potential_bounds,lemma:change in learn,lemma:change in cover} are technical and appear in \cref{sec:deferred} (they are stated for the more general forest case since we will reuse them). Given the bounds above on the different pieces of our potential, we are now ready to prove our main theorem.

\begin{proof}
Let $c(\alg(t))$ denote the cost paid by \cref{alg:nwst} up to and including round $t$.
In every round $t$, the expected cost of the sampled nodes $R^t$ is at most $d^t(s) \cdot \sum_v c_v \cdot x_v^{t-1} / \beta = d^t(s)$ (by \cref{invariant:inner product beta}).
The algorithm pays at most an additional $d^t(s)$ by adding the cheapest path, and hence the total expected cost per round is at most $2 \cdot d^t(s)$.

By combining \cref{lemma:change in learn,lemma:change in cover}, and setting the constants $C_1 = 16$ and $C_2 = 3648e(1+e+e^2+e^3)$, we have 
\begin{align*}
    &\expectover{s^t, R^t}*{\Phi(t) - \Phi(t-1) \,\middle|\, s^1,\ldots,s^{t-1}, R^1,\ldots,R^{t-1}} \nonumber \\
    &\leq - \expectover{s^t, R^t}*{2 \cdot d^t(s) \,\middle|\, s^1,\ldots,s^{t-1}, R^1,\ldots,R^{t-1}},
\end{align*}
which cancels the expected change in $c(\alg(t))$ in each round.
We therefore have the inequality 
\begin{equation}
    \expectover{s^t, R^t}*{\Phi(t) - \Phi(t-1) + c(\alg(t)) - c(\alg(t-1)) \,\middle|\, s^1,\ldots,s^{t-1}, R^1,\ldots,R^{t-1}} \leq 0 \label{eq:cancel cost tree} \text{.}
\end{equation}
By repeatedly applying \eqref{eq:cancel cost tree} for all $t$, we obtain 
\begin{align}
    \expectover{s^t, R^t}*{\Phi(k) - \Phi(0) + c(\alg(k)) - c(\alg(0))} &\leq 0 \nonumber \\
    \expectover{s^t, R^t}*{c(\alg(k))} &\leq \Phi(0) + c(\alg(0)) - \expectover{s^t, R^t}*{\Phi(k)} \nonumber \\
    &\leq O(\beta \cdot \log n) + 0 - (- O(\beta \cdot \log n)) \label{eq:use potential bounds tree},
\end{align}
where \eqref{eq:use potential bounds tree} follows from the fact that $c(\alg(0)) = 0$, together with the bounds on $\Phi$ established in \cref{lemma:potential_bounds}. We conclude that $\expectover{s, R}*{c(\alg(k))} = O(\beta \cdot \log n)$.
\end{proof}

    \section{Random Order Node-Weighted Steiner Forest}

\label{sec:st_forest}

We are ready to describe our fully general algorithm for \nwsf. The primary change is that we replace the naive greedy step with an \emph{augmented greedy} step, as in the work of \cite{DBLP:conf/stoc/BermanC97} for online \ewsf. This complicates the definition of shortcut nodes (because the distance of a terminal to its mate might be at a different scale from its distance to the nearest non-mate terminal), and subsequently to the accounting scheme, but we show how to resolve the issues.

\subsection{The Algorithm}

Once again, we start with a high level description, and note the changes to $\Gamma$ and the addition of $\kappa$ and $T^t_\ell$.
\begin{enumerate}
    \item We maintain a fractional vector $x^t$ as a (not necessarily feasible) solution to the LP relaxation of Steiner Forest, an integral solution $F^t$, and for each particular level $\ell$ (defined below). We also track the set $T^t_\ell$ of terminals at level $\ell$ that arrived before time $t$.  
    \item Every round $t$ when a terminal pair $s^t = (s^t_0, s^t_1)$ arrives (the superscript $t$ is omitted when it is clear from context). 
    \begin{enumerate}
        \item Buy a shortest path, with length $d^t(s) \coloneqq d_{G / F^{t-1}}(s_0, s_1)$, connecting $s_0$ to $s_1$. The \emph{level} of $s$ is defined as $\ell^t_s \coloneqq \lfloor \log_2(d^t(s)) \rfloor$. We then add $s_0$ and $s_1$ to $T_{\ell}^t$.
        \item Sample every vertex $v$ with probability proportional to $x_v$,
        \item Let $\kappa^t_{s_i} \coloneqq d_{G / F^{t-1}}(s_i, T_{\ell}^{t} \setminus \{s_{i}\})$ be the distance of $s_i$ to any other terminal at the same level (including its mate).
        Assuming without loss of generality that $\kappa_{s_0}^t \geq \kappa_{s_1}^t$.
        \begin{itemize}
            \item If $\kappa^t_{s_0} < d^t(s) / 2$, we perform an augmented greedy step, where we connect both terminals to their closest terminals on this level. Note that we do not modify $x^t$ in this case. 
            \item If $\kappa^t_{s_0} \geq d^t(s) / 2$, we increase the value of $x^t_v$ for $v$ in the set
            \[
                \Gamma^t(s) \coloneqq \{
                    v \in V : 
                    d_{G / (F^{t-1} \cup \{v\})}(s_i, T_{\ell}^{t-1}) \leq \kappa^t_{s_i}/2
                \}.
            \]
            We remark that $\Gamma^t(s)$ is only defined when $\kappa^t_{s_0} \geq d^t(s) / 2$. 
        \end{itemize}
    \end{enumerate}
\end{enumerate}
The definition of $d^t(s)$ is the natural generalization of the tree case. However, the new definition of $\Gamma^t(s)$ now restricts attention to terminals on the same level as $s$.

The complete description is in \cref{alg:nwsf}. Like in \cref{sec:st_tree}, let $\Lambda^t$ be the event $X^t(s) \coloneqq \sum_{v \in \Gamma^t(s)} x_v^{t-1} < 1$, so $\neg \Lambda^t$ is the case when our algorithm samples the distribution but does not modify it. We define another event to analyze the augmented greedy step. Let $\Psi^t$ be the event $\kappa^t_{s_0} \geq d^t(s) / 2$, so $\neg \Psi^t$ is the case when our algorithm does the augmented greedy step.

\begin{algorithm}[ht]
\caption{\textsc{LearnOrCoverNWSF}}
\label{alg:nwsf}
\begin{algorithmic}[1]
    \State Initialize $F^0 = \{v \mid c_v < \beta/n\}$.
    \State Initialize $T_{\ell}^0 \leftarrow \emptyset$ \ $\forall \ell \in \mathbb{Z}$.
    \State Set $V \leftarrow \{v \mid \beta/n \leq c_v \leq \beta\}$, and $n' \leftarrow \abs{V}$ .
    \State $x_v^0 \leftarrow \frac{\beta}{c_v n'} \ \forall v \in V,$ 0 otherwise.
    \For{each terminal $s$ in random order}
        \State Set $F^t \leftarrow F^{t-1}$, $R^t \leftarrow \emptyset$, $T_{\ell}^t \leftarrow T_{\ell}^{t-1}$ $\forall \ell \in \mathbb{Z}$.
        \State Set $P$ to be the shortest path in $G / F^{t-1}$ connecting $s_0$ to $s_1$.
        \State Add $F^t \leftarrow F^t \cup P$.
        \State Add $T_{\ell}^t \leftarrow T_{\ell}^t \cup \{s_0,s_1\}$.
        \For{$v \in V$} 
            \State Add $R^t \leftarrow R^t \cup \{v\}$ w.p. $\min \left( d^{t}(s^t) \cdot x_v^{t-1} / \beta, 1  \right)$.
            \State Add $F^t \leftarrow F^t \cup R^t$.
        \EndFor
        \If {$\kappa^t_{s_0} < d^{t}(s) / 2$}~\label{alg:AG} \Comment{$\Psi^t$ does not hold}
            \State Set $P_0$ to be the shortest path in $G / F^{t-1}$ connecting $s_0$ to a terminal in $T_{\ell}^{t-1} \cup \{s_1\}$.
            \State Set $P_1$ to be the shortest path in $G / F^{t-1}$ connecting $s_1$ to a terminal in $T_{\ell}^{t-1} \cup \{s_0\}$.
            \State Add $F^t \leftarrow F^t \cup P_0 \cup P_1$. \label{eq:augmented greedy}
            \State $x^t \leftarrow x^{t-1}$, \Continue 
        \EndIf 
        \If {$\sum_{v \in \Gamma^t(s)} x_v^{t-1}  \geq 1$} \label{line: preconnected} \Comment{$\Lambda^t$ does not hold}
            \State $x^t \leftarrow x^{t-1}$, \Continue 
        \EndIf
        \For{$v \in \Gamma^t(s)$} 
            \State $x_v^t \leftarrow x_v^{t-1} \cdot \exp\{d^{t}(s)/c_v\}$.
        \EndFor
        \For{$v \not \in \Gamma^t(s)$} 
            \State $x_v^t \leftarrow x_v^{t-1}$.
        \EndFor
        \State Set $Z^t \leftarrow \sum_{v \in V} c_v \cdot x_v^t / \beta$ and $x^t \leftarrow x^t / Z^t$. \label{eq:normalize}
    \EndFor
    \State \Return $F^{k}$.
\end{algorithmic}
\end{algorithm}

Note that once again \cref{eq:normalize} ensures that \cref{invariant:inner product beta} holds. We are ready for the analysis. 

\subsection{The Analysis}

For the new potential function, we need to define new and redefine old terms. In particular, the definition of $D^t(s^t)$, the set of vertices claimed by $s^t$, is more involved. For one, $s^t$ is now the terminal \emph{pair} $(s_0^t, s_1^t)$. More substantively, when $\kappa^t_{s_0} \ll d^t(s)$, it is tempting to set the radius of the balls claimed by $D^t(s^t)$ to be on the scale of $\kappa^t_{s_0}$ to ensure disjointness. However this then does not suffice to pay for our algorithm's cost, so instead we sidestep this case and handle it via an auxiliary potential term inspired by \cite{DBLP:conf/stoc/BermanC97}.

Letting $B^t := B_{G / F^{t-1}}(s_0, 2^{\ell-3})$ and $D_{< t} := \bigcup_{\substack{t' < t: \\ \ell^{t'} = \ell}} D^{t'}$, then we say $s$ \emph{claims} the vertex set
\begin{equation}\label{eq:balldef}
    D^t = 
    \begin{cases}
        \displaystyle B^t \setminus D_{<t} & \text{if }\kappa_{s_0}^t \geq d^t(s) / 2 \\
        \emptyset & \text{otherwise}
    \end{cases}
\end{equation}

Said differently, when $\Psi^t$ holds, the incoming terminal pair at level $\ell$ claims the vertices within radius $2^{\ell-3}$, except those that were claimed by vertices at the same level. This refined definition ensures that vertices again claim non-intersecting sets, which again ensures that \eqref{eq:disjoint_charging} holds and the LP cost inside claimed vertices is a good proxy for $\lpopt$. Let us now redefine the potential function $\Phi$.

The main change is that we redefine our learning potential as
\begin{equation*}
    \Phi_{L}(t) = \wKL{x^{*}}{x^t} + 8  
\sum_{\ell=\lfloor\log_2(\frac{\beta}{n})\rfloor}^{\lfloor\log_2(\beta)\rfloor} \left( \lpopt - \sum_{\substack{i \in [t] : \\ \ \ell^i = \ell}} \ \sum_{v \in D^i(s^i)} c_v x_v^* \right)
    + \left( \sum_{\substack{i \in [t] : \\ \Psi^i \text{ true}}} 2^{\ell^i-2} - \hspace{-0.1in}\sum_{\substack{i \in [t]: \\ \Psi^i \text{ false}}} 2^{\ell^i-2} \right) \text{,}
\end{equation*}
where the last term is the promised fix for the case when $\kappa^t_{s_0} \ll d^t(s)$. The covering potential is again 
\begin{equation*}
    \Phi_{C}(t) = \beta \cdot \log \left( \frac{\rho^t}{\beta} + \frac{1}{n} \right),
\end{equation*}
but we also need a more sensitive definition of $\rho^t$. Define
\begin{equation}\label{eq:kappadef}
    \kappa^t_{s} = \kappa^t_{s_0} + \kappa^t_{s_1} + 2^{\ell+3}
\end{equation}
and set $\rho^t = \sum_{s \in U^t} \kappa_{s}^{t+1}$ as the proxy for how much will spend on future requests. Finally, we combine all terms in the single potential 
\begin{equation}\label{eq:potentialdef}
    \Phi(t) = C_1 \cdot \Phi_{L}(t) + C_2 \cdot \Phi_{C}(t) \text{,}
\end{equation}
where $C_1$ and $C_2$ are constants to be specified later.
With all this, we show:

\thmnwsf*

The proof follows the same structure as \cref{sec:st_tree}. We can reuse the initial potential bound of \cref{lemma:potential_bounds} verbatim, again we bound the drop in the learning and covering potentials separately.

\changelearn*

\changecover*

The technical proofs are again deferred to \cref{sec:deferred}, and we conclude with the theorem proof. 

\begin{proof}[Proof of \cref{thm:main_nwsf}]
Let $c(\alg(t))$ be the cost paid by algorithm~\ref{alg:nwsf} up to and including time $t$.
In every round $t$, the expected cost of the sampled nodes $R^t$ is at most $d^t(s) \cdot \sum_v c_v \cdot x_v^{t-1} / \beta = d^t(s)$ (by invariant~\ref{invariant:inner product beta}).
The algorithm pays at most an additional $d^t(s)$ by adding the cheapest path and at most an additional \(d^t(s)\) if the augmented greedy step is being performed.
Hence the total expected cost per round is at most $3 \cdot d^t(s)$.

By combining lemmas \ref{lemma:change in learn} and \ref{lemma:change in cover}, and setting the constants $C_1 = 24$ and $C_2 = 5472e(1+e+e^2+e^3)$, we have 
\begin{align*}
    &\expectover{s^t, R^t}*{\Phi(t) - \Phi(t-1) \,\middle|\, s^1,\ldots,s^{t-1}, R^1,\ldots,R^{t-1}} \nonumber \\
    &\leq - \expectover{s^t, R^t}*{3 \cdot d^t(s) \,\middle|\, s^1,\ldots,s^{t-1}, R^1,\ldots,R^{t-1}} \text{,}
\end{align*}
which cancels the expected change in $c$ in each round.
We therefore have the inequality 
\begin{equation}
    \expectover{s^t, R^t}*{\Phi(t) - \Phi(t-1) + c(\alg(t)) - c(\alg(t-1)) \,\middle|\, s^1,\ldots,s^{t-1}, R^1,\ldots,R^{t-1}} \leq 0 \label{eq:cancel cost} \text{.}
\end{equation}
By repeatedly applying \ref{eq:cancel cost} for all $t$, we obtain 
\begin{align}
    \expectover{s, R}*{\Phi(k) - \Phi(0) + c(\alg(k)) - c(\alg(0))} &\leq 0 \nonumber \\
    \expectover{s, R}*{c(\alg(k))} &\leq \Phi(0) + c(\alg(0)) - \expectover{s, R}*{\Phi(k)} \nonumber \\
    &\leq O(\beta \cdot \log n) + 0 - (- O(\beta \cdot \log n)) \label{eq:use potential bounds} \text{,}
\end{align}
where \ref{eq:use potential bounds} follows from the fact that $c(\alg(0)) = 0$, together with the bounds on $\Phi$ established in lemma~\ref{lemma:potential_bounds}. We conclude that $\expectover{s, R}*{c(\alg(k))} = O(\beta \cdot \log n).$
\end{proof}

    \section{Conclusion}
In this paper, we give an $O(\log n)$-competitive algorithm for random order \nwsf, matching the best possible offline approximation in the $n = \poly(k)$ regime. Run offline, this algorithm recovers a randomized rounding filtering technique that we hope is useful for other problems. 

Beyond our specific results, we hope that this work inspires further research on whether the $\loc$ paradigm can be used for fast \emph{offline} approximation algorithms. More broadly, we hope it demonstrates that designing for restricted computational models can yield insights that are beneficial even to classic offline settings.

    \newpage

    \appendix
    \section{Prize-Collecting Node-Weighted Steiner Forest}\label{sec:pcst_forest}

The prize-collecting node-weighted forest problem (\pcnwsf) 
generalizes the \nwsf problem by an additional penalty term. 
In iteration $t$, a terminal pair $s$ arrives with its penalty $\pi_s$. 
We can either connect the pair $s$, or pay a penalty of $\pi_s$ and skip this iteration. 
Note, since requests can be repeated, it is now possible that \(k = \omega(n)\).
By setting $\pi_s = \infty$, we recovers the \nwsf problem. 
This problem requires a more general LP: 
\begin{alignat}{3}
\min\quad & \sum_{v \in V} c_v x_v + \sum_{s \in T} \pi_s z_s\nonumber \\
\text{s.t.}\quad & z_s + \sum_{P \in \mathcal{P}_{s}} f_P \geq 1, \quad && \forall s \in T \label{eq:pc-path-demand} \\
& \sum_{\substack{P \in \mathcal{P}_{s} \\ P \ni v}} f_P \leq x_v, \quad && \forall s \in T, \, \forall v \in V \label{eq:pc-vertex-capacity} \\
& x_v \geq 0, \quad && \forall v \in V \nonumber \\
& z_s \geq 0, \quad && \forall z \in T \nonumber \\
& f_P \geq 0, \quad && \forall P \in \mathcal{P} \nonumber
\end{alignat}
where variables $z_s$ are indicator variables for whether we choose to pay the penalty for the terminal $s$. 
Observe that even though the same terminal pair may arrive multiple times, we may assume \(\opt\) assigns the same value of \(z_s\) to each of them, and hence we write only one value \(\pi_s\) in the LP above to stand in for the sum of the penalties that appear online.

Our algorithm for \pcnwsf is essentially the same as the algorithm for the \nwsf problem, with the following modifications.
We assume that \(d^t(s) \leq \beta\) for all terminal pairs \(s\), since \(\opt\) must have paid the penalty otherwise. 
In iteration \(t\) with incoming terminal \(s\), if \(\min(d^t(s), \pi_s) \leq \beta/t\), then we either pay the cheap penalty or buy the cheap path, whichever is appropriate.
We introduce the notation \(\Xi^t\) for the event that \(\min(d^t(s), \pi_s) > \beta/t\) and \(\Xi^t(u)\) for the condition that \(\min(d^t(u), \pi_u) > \beta/t\) for arbitrary \(u \in U^{t-1}\).
Otherwise, we first sample each node with probability \(\min\left( \frac{d^{t}(s) \cdot \pi_s}{d^t(s) + \pi_s} \cdot x_v^{t-1}  / \beta, 1  \right)\).
We then do one of two things.
The first is that we pay the penalty with probability \(\frac{d^t(s)}{d^t(s) + \pi_s}\).
Otherwise, we follow the steps of \cref{alg:pcnwsf} from the previous section, running augmented greedy if possible, and if not, we buy a shortest path and update our distribution in the usual way.
Note that \(\Gamma\) still consists of nodes whose purchase would decrease the cost to connect $s_i$ to any terminal in $T_{\ell}^{t-1}$ by a factor of 2 for some $i$ satisfying $\kappa^t_{s_i} \geq d^t(s) / 2$.

\begin{algorithm}
\caption{\textsc{LearnOrCoverPCNWSF}}
\label{alg:pcnwsf}
\begin{algorithmic}[1]
    \State Initialize $F^0 = \{v \mid c_v < \beta/n\}$.
    \State Initialize  $T_{\ell}^0 \leftarrow \emptyset$ \ $\forall \ell \in \mathbb{Z}$.
    \State Set $V \leftarrow \{v \mid \beta/n \leq c_v \leq \beta\}$, and $n' \leftarrow \abs{V}$.
    \State $x_v^0 \leftarrow \frac{\beta}{c_v n'} \ \forall v \in V,$ 0 otherwise.
    \For{each request $s$ at time $t$ in random order}
        \State Set $F^t \leftarrow F^{t-1}$, $R^t \leftarrow \emptyset$, $T_{\ell}^t \leftarrow T_{\ell}^{t-1}$ $\forall \ell \in \mathbb{Z}$.
        \State Set $P$ to be the shortest path in $G / F^{t-1}$ connecting $s_0^t$ to $s_1^t$.
        \If {\(\pi_s \leq \frac{\beta}{t}\)}
            \State Pay penalty \(\pi_s\), \Continue ($\Xi^t$ does not hold)
        \EndIf
        \If {\(d^t(s) \leq \frac{\beta}{t}\)}
            \State Add $F^t \leftarrow F^t \cup P$, \Continue ($\Xi^t$ does not hold)
        \EndIf
        \For{$v \in V$} 
            \State Add $R^t \leftarrow R^t \cup \{v\}$ w.p. $\min \left( \frac{d^{t}(s) \cdot \pi_s}{d^t(s) + \pi_s} \cdot x_v^{t-1}  / \beta, 1  \right)$.
            \State Add $F^t \leftarrow F^t \cup R^t$.
        \EndFor
        \State Sample \(Y^t \sim \Ber\left( \frac{d^t(s)}{d^t(s) + p_{s}} \right)\). \label{eq:penalty coin pc}
        \If{ \( Y^t = 1 \) }
            \State Pay penalty \(p_{s}\). 
            \State $x^t \leftarrow x^{t-1}$, \Continue
        \EndIf
        \State Set $F^t \leftarrow F^t \cup P$.
        \State Set $T_{\ell}^t \leftarrow T_{\ell}^t \cup \{s_0,s_1\}$.
        \If {$\kappa^t_{s_0} < d^{t}(s) / 2$}~\label{alg:AG pc}
            \State Set $T_{\ell}^t \leftarrow T_{\ell}^t \cup \{s_0,s_1\}$.
            \State Set $P_0$ to be the shortest path in $G / F^{t-1}$ connecting $s_0$ to a terminal in $T_{\ell}^{t-1} \cup \{s_1\}$.
            \State Set $P_1$ to be the shortest path in $G / F^{t-1}$ connecting $s_1$ to a terminal in $T_{\ell}^{t-1} \cup \{s_0\}$.
            \State Add $F^t \leftarrow F^t \cup P_0 \cup P_1$. \label{eq:augmented greedy pc}
            \State $x^t \leftarrow x^{t-1}$, \Continue ($\Psi^t$ does not hold)
        \EndIf
        \If {$\sum_{v \in \Gamma^t(s)} x_v^{t-1}  \geq 1$} \label{line: preconnected pc}
            \State $x^t \leftarrow x^{t-1}$, \Continue ($\Lambda^t$ does not hold)
        \EndIf
        \For{$v \in \Gamma^t(s)$} 
            \State $x_v^t \leftarrow x_v^{t-1} \cdot \exp\left\{d^{t}(s) /c_v\right\}$.
        \EndFor
        \For{$v \not \in \Gamma^t(s)$} 
            \State $x_v^t \leftarrow x_v^{t-1}$.
        \EndFor
        \State Set $Z^t \leftarrow \sum_{v \in V} c_v \cdot x_v^t / \beta$ and $x^t \leftarrow x^t / Z^t$. \label{eq:normalize pc}
    \EndFor
    \State \Return $F^{k}$.
\end{algorithmic}
\end{algorithm}

To analyze \cref{alg:pcnwsf}, we modify a few definitions from \cref{sec:st_forest}.
Let $(x^*, z^*)$ be an optimal fractional solution to the above LP. 
We define the sets claimed by incoming request $s^t$ as
\begin{equation}\label{eq:balldef pc}
    D^t = 
    \begin{cases}
        \displaystyle B_{G / F^{t-1}}(s_0, 2^{\ell-3}) \setminus \bigcup_{\substack{t' < t: \\ \ell^{t'} = \ell}} D^{t'} & \kappa_{s_0}^t \geq d^t(s) / 2 \text{ and } Y^t = 0 \\
        \emptyset & \text{otherwise}
    \end{cases}
\end{equation}
This new definition is so that we only claim vertices when we buy a path, ensuring that these regions behave the same way as in previous sections.
We then define our learning potential as
\begin{align*}
    \Phi_{L}(t) &= \wKL{x^{*}}{x^t} + 8 \cdot \sum_{\ell=0}^{\log_2 k} \left( \lpopt - \sum_{\substack{i \in [t] : \\ \ \ell^i = \ell}} \ \sum_{v \in D^i(s^i)} c_v \cdot x_v^* \right) + \sum_{s \in U^t} \pi_s \cdot z_s^*\\
    &+ \left( \sum_{i \in [t] : \Psi^i \land Y^i=0} 2^{\ell^i-2} - \sum_{i \in [t] : \neg \Psi^i \land Y^i=0} 2^{\ell^i-2} \right) + 3\cdot(e^4+5) \sum_{i=t+1}^k \frac{\beta}{i} \text{.}
\end{align*}
The new term on the first line of the learning potential allows us to charge some costs to \(\opt\) when it pays the penalty.
The new final term of the learning potential simply exists to account for steps in which \(\min(d^t(s), \pi_s) \leq \beta/t\).

Before we define our covering potential, we define \(\kappa_s^t\) as
\begin{equation}\label{eq:kappadef pc}
    \kappa^t_{s} = \begin{cases}
        \frac{\pi_s}{\log_2 \left( \frac{38}{37} \right)} \cdot \log_2 \left( \frac{\kappa^t_{s_0} + \kappa^t_{s_1} + 2^{\ell+3}}{\pi_s} \right) & d^t(s) > \pi_s \\
        \kappa^t_{s_0} + \kappa^t_{s_1} + 2^{\ell+3} & \text{otherwise}
    \end{cases} 
\end{equation}
The intuition behind this new definition when \(d^t(s) > \pi_s\) is that we want \(\kappa^t_s\) to decrease by \(\Omega(\pi_s)\) when a node is sampled from \(\Gamma^t(s)\) since this is the expected amount we spend on steps where augmented greedy doesn't occur.
However, since it is possible that \(\frac{d^t(s)}{\pi_s} = k\), we may need to sample from \(\Gamma(s)\) at \(O(\log k)\) different timesteps before \(d^t(s) \leq \pi_s\).
Hence we multiply \(\pi_s\) by a multiplicative factor proportional to \(\log_2 \left( \frac{d^t(s)}{\pi_s} \right)\) so that \(\kappa^t_{s}\) is large enough to decrease by \(\Omega(\pi_s)\) a sufficient number of times.
To compensate for this extra factor compared to the previous sections, we define the covering potential as 
\begin{equation*}
     \Phi_{C}(t) = \beta \cdot \log_2 (10k) \cdot \log \left( \frac{\rho^t}{\beta} + \frac{1}{k} \right)\text{.}
\end{equation*}
Finally, we as usual set
\begin{equation*}
    \Phi(t) = C_1 \cdot \Phi_L(t) + C_2 \cdot \Phi_C(t) \text{,}
\end{equation*}
where $C_1, C_2$ are constants to be specified later. 

We now begin our analysis, where the full proofs of most lemmas may be found in \cref{sec:deferred}.

\begin{restatable}{lemma}{pcpotentialbounds}
    \label{lemma:pc potential bounds}
    The initial potential is bounded as $\Phi(0) = O(\beta \cdot (\log n + \log^2 k))$, and $\Phi(t) \geq -O( \beta \cdot \log^2 k)$ for all $t$.
\end{restatable} 

It is readily apparent that the new terms in the learning potential are insignificant with regards to these bounds, and the new \(\log^2k\) bounds are a consequence of the extra \(\log k\) factor in the covering potential.

\begin{restatable}{lemma}{pcchangelearn}
    \label{lemma:pc change in learn}
    For all rounds \(t\) in which \(\Xi^t\) holds, the expected change in $\Phi_L$ is 
    \begin{align}
        &\expectover{s^t, R^t}*{\Phi_L(t) - \Phi_L(t-1) \,\middle|\, x^{t-1}, F^{t-1}, U^{t-1}, \Xi^t} \nonumber \\
        &\leq \expectover{s \sim U^{t-1}}*{\frac{e^{4} - 1}{4} \cdot \frac{d^t(s) \cdot \pi_s}{d^t(s) + \pi_s} \cdot \min \left( X^t(s), 1 \right) - \frac{d^t(s) \cdot \pi_s}{8(d^t(s) + \pi_s)} \,\middle|\, x^{t-1}, F^{t-1}, \Xi^t} \text{.} \label{eq:pc learn bound 1}
    \end{align}
    In rounds in which \(\neg \Xi^t\) holds, the expected change in \(\Phi_L\) is
    \begin{equation}
        \expectover{s^t, R^t}*{\Phi_L(t) - \Phi_L(t-1) \,\middle|\, x^{t-1}, F^{t-1}, U^{t-1}, \neg \Xi^t} \leq -3\cdot(e^4+5) \cdot \frac{\beta}{t}\text{.} \label{eq:pc learn bound 2}
    \end{equation}
\end{restatable}

This change in learn bound uses essentially the same arguments from the previous sections.
We first observe that if \(Y^t = 1\), then \(x^t = x^{t-1}\) and \(D^t = \emptyset\), so our potential decreases by at least \(\pi_s \cdot z^*_s\).
To handle the case where \(Y^t = 0\), we first introduce the following lemma which is a slight variation of \cref{lemma:alternate}:
\begin{restatable}{lemma}{pc alternate}
    \label{lemma:pc-alternate}
    Let $s \in T$ and $F \subseteq V$.
    Let $\gamma < d_{G / F}(s)$. For any cover $A \cup B \supseteq B_{G / F}(s_i,\gamma) \eqqcolon B(s_i,\gamma)$, 
    \begin{equation}\label{eq:pc weight in cover}
        \sum_{v \in A} c_v x_v^* \geq \frac{\gamma}{2}(1 - z^*_s) \quad \text{or} \quad \sum_{v \in B} c_v x_v^* \geq \frac{\gamma}{2}(1 - z^*_s).
    \end{equation}
    Moreover, the statement still holds if we replace any of the node weights $c_u$ with $c_u' \geq \gamma$.
\end{restatable}
With this and the bounds from \cref{sec:st_forest}, we know that the change in learning potential in this case is bounded above by \(\frac{e^{4} - 1}{4} \cdot \frac{d^t(s) \cdot \pi_s}{d^t(s) + \pi_s} \cdot \min \left( X^t(s), 1 \right) - \frac{d^t(s) \cdot \pi_s}{8(d^t(s) + \pi_s)} \cdot (1 - z^*_s)\).
The law of total expectation completes the proof.

\begin{restatable}{lemma}{pcchangecover}
    \label{lemma:pc change in cover}
    For all rounds $t$ in which \(\Xi^t\) holds, the expected change in $\Phi_C$ is 
    \begin{align}\label{eq:change in cover not preconnected pc}
        &\expectover{s^t, R^t}*{\Phi_C(t) - \Phi_C(t-1) \,\middle|\, x^{t-1}, F^{t-1}, U^{t-1}, \Xi^t} \\
        &\leq -\frac{1 - e^{-1}}{38}  \log_2 (10k) \frac{\sum_{s \in U^{t-1}} \mathbbm{1}(\Xi^t(s)) \frac{d^t(s) \cdot \pi_s}{d^t(s) + \pi_s}}{\rho^{t-1} + \frac{\beta}{k}}  \expectover{u \sim U^{t-1}}*{\frac{d^t(u) \cdot \pi_s}{d^t(u) + \pi_u}  \min(X^t(u), 1) \,\middle|\, x^{t-1}, F^{t-1}, \Xi^t} \text{.}
    \end{align}
    In rounds in which \(\neg \Xi^t\) holds, the expected change in \(\Phi_C\) is
    \begin{align}\label{eq:change in cover preconnected pc}
        &\expectover{s^t, R^t}*{\Phi_C(t) - \Phi_C(t-1) \,\middle|\, x^{t-1}, F^{t-1}, U^{t-1}, \neg \Xi^t} \leq 0 \text{.}
    \end{align} 
\end{restatable}

This is again obtained by using nearly identical arguments as those from previous sections.
However, this inequality was able to be simplified further previously due to \(\Xi^t\) holding for all \(s \in U^{t-1}\). 

\begin{restatable}{theorem}{thmpcnwst}
    \label{thm:main_pcnwst}
    There is a polynomial-time $O(\log(n) + \log^2 k)$-competitive algorithm for random order online \pcnwsf.
\end{restatable}

\begin{proof}
Let $c(\alg(t))$ be the cost paid by \cref{alg:nwsf} up to and including time $t$.
In every round $t$ where \(\Xi^t\) holds, the expected cost of the sampled nodes $R^t$ is at most $\frac{d^t(s)\cdot \pi_s}{d^t(s)+ \pi_s} \cdot \sum_v c_v \cdot x_v^{t-1} / \beta = \frac{d^t(s)\cdot \pi_s}{d^t(s)+ \pi_s}$ (by \cref{invariant:inner product beta}).
In these rounds, the algorithm also pays in expectation at most an additional \(4 \cdot \frac{d^t(s)\cdot \pi_s}{d^t(s)+ \pi_s}\).
Hence the total expected cost in such rounds is at most $5 \cdot  \frac{d^t(s)\cdot \pi_s}{d^t(s)+ \pi_s}$.

First assume that 
\begin{equation}\label{eq:pc assumption not preconnected}
    \sum_{\substack{s \in U^{t-1}: \\ \Xi^t(s)}} \frac{d^t(s) \cdot \pi_s}{d^t(s) + \pi_s} \geq \sum_{\substack{s \in U^{t-1}: \\ \neg \Xi^t(s)}} \frac{d^t(s) \cdot \pi_s}{d^t(s) + \pi_s} \text{.}
\end{equation}
Then since
\begin{align*}
    \frac{d^t(s) \cdot \pi_s}{d^t(s) + \pi_s} \cdot \log_2 (10k)   &\geq \frac{\pi_s}{2} \cdot \log_2 (10k)  \\
    &\geq \frac{\pi_s}{2} \cdot \log_2 \left(10 \cdot \frac{d^t(s)}{\pi_s} \right) \\
    &\geq \frac{\log_2\left( \frac{38}{37} \right)}{2} \cdot \frac{\pi_s}{\log_2\left( \frac{38}{37} \right)}\cdot \log_2 \left( 10 \cdot \frac{d^t(s)}{\pi_s} \right)\\
    &\geq \frac{\log_2\left( \frac{38}{37} \right)}{2} \cdot \frac{\pi_s}{\log_2 \left( \frac{38}{37} \right)} \cdot \log_2 \left( \frac{\kappa^t_{s_0} + \kappa^t_{s_1} + 2^{\ell+3}}{\pi_s} \right) \\
    &= \frac{\log_2\left( \frac{38}{37} \right)}{2} \cdot \kappa_s^t \text{,}
\end{align*}
we have by \cref{lemma:pc change in cover}
\begin{align*}
    &\expectover{s^t, R^t}*{\Phi_C(t) - \Phi_C(t-1) \,\middle|\, x^{t-1}, F^{t-1}, U^{t-1}, \Xi^t} \\
    &\leq -\frac{1 - e^{-1}}{38} \cdot \log_2 (10k) \cdot \frac{\sum_{s \in U^{t-1}:\Xi^t(s)} \frac{d^t(s) \cdot \pi_s}{d^t(s) + \pi_s}}{\rho^{t-1} + \frac{\beta}{k}} \cdot \expectover{u \sim U^{t-1}}*{\frac{d^t(u) \cdot \pi_s}{d^t(u) + \pi_u} \cdot \min(X^t(u), 1) \,\middle|\, x^{t-1}, F^{t-1}, \Xi^t} \\
    &\leq -\frac{1 - e^{-1}}{76} \cdot \frac{\sum_{s \in U^{t-1}} \frac{d^t(s) \cdot \pi_s}{d^t(s) + \pi_s}}{\rho^{t-1} + \frac{\beta}{k}} \cdot \expectover{u \sim U^{t-1}}*{\frac{d^t(u) \cdot \pi_s}{d^t(u) + \pi_u} \cdot \min(X^t(u), 1) \,\middle|\, x^{t-1}, F^{t-1}, \Xi^t} \\
    &\leq -\frac{1 - e^{-1}}{152} \cdot \log_2\left( 38/37 \right) \cdot \frac{\sum_{s \in U^{t-1}} \kappa_s^t}{\rho^{t-1} + \frac{\beta}{k}} \cdot \expectover{u \sim U^{t-1}}*{\frac{d^t(u) \cdot \pi_s}{d^t(u) + \pi_u} \cdot \min(X^t(u), 1) \,\middle|\, x^{t-1}, F^{t-1}, \Xi^t} \\
    &\leq -\frac{1 - e^{-1}}{304} \cdot \log_2\left( 38/37 \right) \cdot \expectover{u \sim U^{t-1}}*{\frac{d^t(u) \cdot \pi_s}{d^t(u) + \pi_u} \cdot \min(X^t(u), 1) \,\middle|\, x^{t-1}, F^{t-1}, \Xi^t} \text{,}
\end{align*}
where the last line follows as \(\kappa_u^t \geq \beta/k\).
By taking \(C_1 = 48\) and \(C_2 = 3648e(1+e+e^2+e^3)/\log_2(38/37)\), we have that under \eqref{eq:pc assumption not preconnected},
\begin{equation}
    \expectover{s^t, R^t}*{\Phi(t) - \Phi(t-1) + c(\alg(t)) - c(\alg(t-1)) \,\middle|\, s^1,\ldots,s^{t-1}, R^1,\ldots,R^{t-1}} \leq 0 \label{eq:cancel cost 1} \text{.}
\end{equation}

Now assume that
\begin{equation}\label{eq:pc assumption preconnected}
    \sum_{s \in U^{t-1}:\Xi^t(s)} \frac{d^t(s) \cdot \pi_s}{d^t(s) + \pi_s} < \sum_{s \in U^{t-1}:\neg \Xi^t(s)} \frac{d^t(s) \cdot \pi_s}{d^t(s) + \pi_s} \text{.}
\end{equation}
Then
\begin{align*}
    &\expectover{s^t, R^t}*{c(\alg(t)) - c(\alg(t-1)) \,\middle|\, s^1,\ldots,s^{t-1}, R^1,\ldots,R^{t-1}} \\
    &\leq \frac{1}{\abs{U^{t-1}}} \left( \sum_{s \in U^{t-1} : \Xi^t(s)} 5 \cdot \frac{d^t(s)\cdot \pi_s}{d^t(s)+ \pi_s} + \sum_{s \in U^{t-1} : \neg\Xi^t(s)} \frac{\beta}{t} \right) \\
    &\leq \frac{1}{\abs{U^{t-1}}} \left( \sum_{s \in U^{t-1} : \neg \Xi^t(s)} 5 \cdot \frac{d^t(s)\cdot \pi_s}{d^t(s)+ \pi_s} + \sum_{s \in U^{t-1} : \neg\Xi^t(s)} \frac{\beta}{t} \right) \\
    &\leq\frac{1}{\abs{U^{t-1}}} \left( \sum_{s \in U^{t-1}} 6 \cdot \frac{\beta}{t} \right) \\
    &= 6 \cdot \frac{\beta}{t} \text{.}
\end{align*}
Furthermore, since \(\frac{d^t(s) \cdot \pi_s}{d^t(s) + \pi_s} \leq \pi_s \leq \beta/t\) if \(\neg \Xi^t(s)\) and \(\frac{d^t(s) \cdot \pi_s}{d^t(s) + \pi_s} \geq \frac{\pi_s}{2} > \frac{\beta}{2t}\) if \(\Xi^t(s)\), \eqref{eq:pc assumption preconnected} implies \(\Pr[\neg \Xi^t] \geq 1/3\), and hence 
\begin{align*}
    &\expectover{s^t, R^t}*{\Phi_L(t) - \Phi_L(t-1) \,\middle|\, x^{t-1}, F^{t-1}, U^{t-1}} \\
    &\leq \Pr[\neg \Xi^t] \cdot \expectover{s^t, R^t}*{\Phi_L(t) - \Phi_L(t-1) \,\middle|\, x^{t-1}, F^{t-1}, U^{t-1}, \neg \Xi^t} \\
    &+ \Pr[\Xi^t] \cdot \expectover{s^t, R^t}*{\Phi_L(t) - \Phi_L(t-1) \,\middle|\, x^{t-1}, F^{t-1}, U^{t-1}, \Xi^t} \\
    &\leq - \Pr[\neg \Xi^t] \cdot 3\cdot(e^4+5) \cdot \frac{\beta}{t} + \frac{1}{\abs{\{s \in U^{t-1} : \Xi^t(s)\}}} \cdot \frac{e^4 - 1}{4} \cdot \sum_{s \in U^{t-1}:\Xi^t(s)} \frac{d^t(s) \cdot \pi_s}{d^t(s) + \pi_s} \\
    &\leq -(e^4+5) \cdot \frac{\beta}{t} + \frac{3}{\abs{U^{t-1}}} \cdot \frac{e^4 - 1}{4} \cdot \sum_{s \in U^{t-1}:\neg \Xi^t(s)} \frac{d^t(s) \cdot \pi_s}{d^t(s) + \pi_s} \\
    &\leq -(e^4+5) \cdot \frac{\beta}{t} + \frac{3}{\abs{\{s \in U^{t-1} : \neg \Xi^t(s)\}}} \cdot \frac{e^4 - 1}{4} \cdot \sum_{s \in U^{t-1}: \neg \Xi^t(s)} \frac{d^t(s) \cdot \pi_s}{d^t(s) + \pi_s} \\
    &\leq -(e^4+5) \cdot \frac{\beta}{t} + (e^4 - 1) \cdot \frac{\beta}{t} \\
    &\leq -6 \cdot \frac{\beta}{t} \text{.}
\end{align*}
Since \(\expectover{s^t, R^t}*{\Phi_C(t) - \Phi_C(t-1) \,\middle|\, x^{t-1}, F^{t-1}, U^{t-1}} \leq 0\), we have that under \eqref{eq:pc assumption preconnected},
\begin{equation}
    \expectover{s^t, R^t}*{\Phi(t) - \Phi(t-1) + c(\alg(t)) - c(\alg(t-1)) \,\middle|\, s^1,\ldots,s^{t-1}, R^1,\ldots,R^{t-1}} \leq 0 \label{eq:cancel cost 2} \text{.}
\end{equation}
By repeatedly applying \ref{eq:cancel cost 1} or \ref{eq:cancel cost 2} for all $t$, we obtain 
\begin{align}
    \expectover{s, R}*{\Phi(k) - \Phi(0) + c(\alg(k)) - c(\alg(0))} &\leq 0 \nonumber \\
    \expectover{s, R}*{c(\alg(k))} &\leq \Phi(0) + c(\alg(0)) - \expectover{s, R}*{\Phi(k)} \nonumber \\
    &\leq  O(\beta \cdot (\log n + \log^2 k)) + 0 - (-O( \beta \cdot \log^2 k)) \label{eq:pc use potential bounds} \text{,}
\end{align}
where \ref{eq:pc use potential bounds} follows from the fact that $c(\alg(0)) = 0$, together with the bounds on $\Phi$ established in \cref{lemma:pc potential bounds}. We conclude that 
$\expectover{s, R}*{c(\alg(k))} = O(\beta \cdot (\log n + \log^2 k))$.
\end{proof}
    \section{Deferred Proofs}\label{sec:deferred}

\subsection{Proofs from \cref{sec:st_tree,sec:st_forest}}

\support*

\begin{proof}
    We assume for contradiction there is some $u$ such that $c_u > \lpopt$ where $x^*_u > 0$. Let $d = x^*_u$ for convenience. Since the total cost of \opt is $\sum_v c_vx^*_v = \lpopt$, then $d$ must be less than 1. Furthermore, $\sum_{v \neq u} c_vx^*_v < \lpopt(1-d)$.  If we create a new solution $x'_v = x^*_v/(1-d)$ for all $v$ except $x'_u$ which we set to be 0. Any terminal pair \(s \in T\) must be fulfilled to degree at least $1-d$ without $u$, otherwise the original solution would not be feasible, so scaling up all other paths by $\frac{1}{1-d}$ will make sure all pairs are satisfied. Furthermore, 
\begin{equation*}
    \frac{1}{1-d} \cdot \sum_{v \neq u} c_vx^*_v < \lpopt(1-d) \cdot \frac{1}{1-d} = \lpopt    
\end{equation*}
So our solution is strictly better than \opt, which is a contradiction.
\end{proof}

We require a structural lemma, which formalizes the intuition that either \opt fractionally uses a significant number of nodes in $\Gamma^t(s)$, or $D^t(s)$ contains sufficient LP cost to pay for a greedy path.

\begin{restatable}[Ball Covering Lemma]{lemma}{alternate}
    \label{lemma:alternate}
    Let $s = (s_0, s_1) \in T$ be a terminal pair, $F \subseteq V$ be our current solution, and let 
    $\gamma < d^t(s)$ 
    be a constant.
    If $A, B \subseteq V$ covers the ball $B_{G / F}(s_i, \gamma)$, i.e. $B_{G / F}(s_i,\gamma) \subseteq A \cup B $, then
    \begin{equation}\label{eq:weight in cover}
        \sum_{v \in A} c_v x_v^* \geq \frac{\gamma}{2} \quad \text{or} \quad \sum_{v \in B} c_v x_v^* \geq \frac{\gamma}{2}.
    \end{equation}
    Moreover, the statement still holds if we replace any of the node weights $c_u$ with $c_u' \geq \gamma$.
\end{restatable}

\begin{proof}
Since $\gamma < d^t(s)$, every path $P \in \mathcal{P}_{s}$ contains vertices in $B(s_i, \gamma)$ whose total cost is more than $\gamma$.
Formally,
\begin{equation}\label{eq:pathsum}
    \sum_{v \in P \cap B(s_i,\gamma)} c_v \geq \gamma \text{.}
\end{equation}
 
We then multiply by $f_P$ and sum over all paths. From our LP, we know from constraint~\eqref{eq:path-demand} that $\sum f_P \geq 1$, so
\begin{equation*}
    \sum_{P \in \mathcal{P}_{s}} f_P \left(\sum_{v \in P \cap B(s_i,\gamma)} c_v\right) \geq \gamma \sum_P f_P \geq \gamma.
\end{equation*}
We rearrange the double sum, and apply the LP constraint~\eqref{eq:vertex-capacity} which states $\sum_{P \ni v} f_P \leq x_v^*$ so that
\begin{equation*}
    \gamma \leq\sum_{P \in \mathcal{P}_{s}}\ \sum_{v \in P \cap B(s_i,\gamma)} f_P \cdot c_v =\sum_{v \in B(s_i,\gamma)} c_v \left(\sum_{\substack{P \ni v}} f_P\right) \leq \sum_{v \in B(s_i,\gamma)} c_v x_v^* \text{.}
\end{equation*}
Finally, we can split the sum into the nodes in each set in the cover. For any cover $A \cup B \supseteq B(s_i,\gamma)$,
\begin{equation*}
    \underbrace{\sum_{v \in A} c_v x_v^*}_{(A)} + \underbrace{\sum_{v \in B\backslash A} c_v x_v^*}_{(B)} \geq \gamma.
\end{equation*}
By an averaging argument, since $(A) + (B) \geq \gamma$
\begin{equation*}
    \sum_{v \in A} c_v x_v^*  \geq \gamma/2 \quad
    \text{or} \quad
    \sum_{v \in B} c_v x_v^* \geq \sum_{v \in B\backslash A} c_v x_v^*\geq \gamma/2 \text{.} 
\end{equation*}
This proves \eqref{eq:weight in cover}.
If we do replace any of the node weights with $c'_u \geq \gamma$, the only step this affects is \eqref{eq:pathsum}, as every other step only relies on rearranging the sums. 
Looking at \eqref{eq:pathsum}, we see that if we replace $c_u$ with $c_u' \geq \gamma$, the inequality still holds, as any sum containing $c_u$ follows $\sum_v c_v \geq  c_u' \geq \gamma$, and fulfills the inequality on its own. 
Every other step follows exactly in the same manner.
\end{proof}

\potentialbounds*

\begin{proof}
We know that \(\supp{x^*} \subseteq \supp{x^0}\) by \cref{fact:support}.
We then bound the initial KL-divergence term by 
\begin{align*}
    \wKL{x^*}{x^0} &= \sum_{v} c_v \cdot x_v^* \cdot \log \left( x_v^* \frac{c_v n'}{\beta} \right) + \sum_v c_v (x_v^0 - x_v^*) \\
    &\leq \sum_{v} c_v \cdot x_v^* \cdot \log \left( n' \right) + \sum_v c_v \cdot x_v^0 - \sum_v c_v \cdot x_v^* \\
    &\leq \beta(\log n + 1) \text{,}
\end{align*}
where we used \(c_v \cdot x_v^* \leq \beta\) and \(\sum_v c_v \cdot x_v^* = \lpopt \leq \beta = \sum_v c_v \cdot x_v^0\).

Turning to the second part of the \(\Phi_L(0)\) term, we compute
\begin{equation*}
    \sum_{\ell=0}^{\log_2 n} \left( \lpopt - \sum_{\substack{i \in [t] : \\ \ell^i = \ell}} \sum_{v \in D^i(s^i)} c_v \cdot x_v^* \right) = (\log_2 n+1) \cdot \lpopt =  O(\beta \cdot \log n) \text{.}
\end{equation*}

The third part of the \(\Phi_L(0)\) term is trivially 0.
We now turn to the cover term.
We observe
\begin{equation*}
    \kappa_{s}^0 < 2^{\ell_{s}^0 + 2} + 2^{\ell_{s}^0 + 3} \leq 12 \cdot d^0(s) \text{,}
\end{equation*}
by our definition of \(\ell_{s^t}^0\) and \(\kappa_{s^t}^0\).
Since \(d_{s^t}^0 \leq \beta\) for all clients \(s^t\),
\begin{equation*}
    \beta \cdot \log \left( \frac{\rho^t}{\beta} + \frac{1}{n} \right) = O(\beta \cdot \log k)  = O(\beta \cdot \log n)  \text{.}
\end{equation*}
We conclude with the lower bound in \(\Phi\).
The first two parts in the learning term of \(\Phi\) are nonnegative since the \(D^i\) are disjoint when fixing \(\ell\) by \cref{eq:balldef}, and so if we can show that the third part is nonegative, it will follow that
\begin{equation*}
    \Phi(t) \geq C_2 \cdot \beta \cdot \log \left( \frac{\rho^t}{\beta} + \frac{1}{n} \right) \geq -O( \beta \cdot \log n) \text{,} 
\end{equation*}
since \(\rho^t \geq 0\).
To show that the third part of the \(\Phi_L\) term is nonnegative, it is enough to show that
\begin{equation}\label{eq:appendix potential bound explanation}
    \sum_{i \in [k] : \Psi^i} 2^{\ell^i} \geq \sum_{i \in [k] : \neg \Psi^i} 2^{\ell^i} \text{.}
\end{equation}
Fix \(\ell\).
If \(\neg \Psi^t\) holds for some \(t \in [k]\) with \(\ell^t = \ell\), we perform the augmented greedy step.
In this case, since \(d^t(s)/2 > \kappa_{s_0}, \kappa_{s_1}\), the number of components containing terminals from \(T^{t}_{\ell}\) in \(F^t\) must decrease compared to the number of components containing terminals from \(T^{t-1}_{\ell}\) in \(F^{t-1}\).
If \(\Psi^{t}\) holds for some \(t \in [k]\) with \(\ell^{t} = \ell\), the number of components containing terminals from \(T^{t}_{\ell}\) in \(F^{t}\) can increase by at most one compared to the number of components containing terminals from \(T^{t-1}_{\ell}\) in \(F^{t-1}\).
In steps where \(\ell^{t} \neq \ell\), we observe \(T^{t}_{\ell}\) = \(T^{t-1}_{\ell}\), and hence the number of components containing such vertices does not increase as we cannot return purchased nodes.
Inequality~\eqref{eq:appendix potential bound explanation} follows.
\end{proof}

\changelearn*

\begin{proof}
Inequality~\eqref{eq:change in learn not preconnected} is straightforward if \(\neg \Psi^t\) holds.
Since \(x^t = x^{t-1}\) the KL-divergence part of \(\Phi_L\) is trivial, and since \(D^t(s) = \emptyset\), this second part is also trivial.
Since \(2^{\ell-2} \geq d^t(s)/8\), the inequality follows

We now prove \ref{eq:change in learn not preconnected} by breaking into cases when \(\Psi^t\) holds.
If \(\Lambda^t\) does not hold $(X^t(s) \geq 1)$, then \(x^t = x^{t-1}\), and since \((e^4 - 1) / 4 > 1\), inequality~\eqref{eq:change in learn not preconnected} holds.
Hence we will assume \(\Lambda^t\) holds, and the law of total expectation will complete the proof.

Expanding definitions, we find 
\begin{align}
    &\expectover{s^t, R^t}*{\wKL{x^{*}}{x^t} - \wKL{x^{*}}{x^{t-1}} \,\middle|\, x^{t-1}, F^{t-1}, U^{t-1}, \Psi^t, \Lambda^t} \nonumber \\
    &= \expectover{s \sim U^{t-1},R^t}*{\wKL{x^{*}}{x^t} - \wKL{x^{*}}{x^{t-1}} \,\middle|\, x^{t-1}, F^{t-1}, \Psi^t, \Lambda^t} \nonumber \\
    &= \expectover{s \sim U^{t-1}}*{\sum_v c_v \cdot x_v^* \cdot \log \frac{x_v^{t-1}}{x_v^t} \,\middle|\, x^{t-1}, F^{t-1}, \Psi^t, \Lambda^t} \nonumber \\
    &= \expectover{s \sim U^{t-1}}*{\sum_{v \not \in \Gamma^{t}(s)} c_v \cdot x_v^* \cdot \log Z^t + \sum_{v \in \Gamma^{t}(s)} c_v \cdot x_v^* \cdot \log \frac{Z^t}{\exp(d^t(s) / c_v)} \,\middle|\, x^{t-1}, F^{t-1}, \Psi^t, \Lambda^t} \nonumber \\
    &= \expectover{s \sim U^{t-1}}*{ \sum_v c_v \cdot x_v^* \cdot \log Z^t - \sum_{v \in \Gamma^{t}(s)} c_v \cdot x_v^* \cdot \log \exp(d^t(s) / c_v) \,\middle|\, x^{t-1}, F^{t-1}, \Psi^t, \Lambda^t} \nonumber \\
    &\leq \expectover{s \sim U^{t-1}}*{ \beta \cdot \log Z^t - \sum_{v \in \Gamma^{t}(s)} x_v^* \cdot d^t(s) \,\middle|\, x^{t-1}, F^{t-1}, \Psi^t, \Lambda^t} \label{eq:opt bound} \\
    &\leq \expectover{s \sim U^{t-1}}*{
    \begin{array}{ll}
        \displaystyle \beta \cdot \log \left( \sum_{v \in \Gamma^{t}(s)} \frac{c_v}{\beta} \cdot x_v^{t-1} \cdot \exp(d^t(s) / c_v) + \sum_{v \not \in \Gamma^{t}(v)} \frac{c_v}{\beta} \cdot x_v^{t-1}   \right) \\
        - \displaystyle \sum_{v \in \Gamma^{t}(s)} x_v^* \cdot d^t(s)
    \end{array}
    \,\middle|\, x^{t-1}, F^{t-1}, \Psi^t, \Lambda^t
    } \label{eq:expand Z} \text{,}
\end{align}
where step~\eqref{eq:opt bound} follows from \(\sum c_v \cdot x_v^* \leq \beta\) and step~\eqref{eq:expand Z} follows by the definition of \(Z^t\).
We can then further bound \ref{eq:expand Z} by 
\begin{equation}
    \leq \expectover{s \sim U^{t-1}}*{ \beta \cdot \log \left( \sum_{v} \frac{c_v}{\beta} \cdot x_v^{t-1} + \frac{e^4 - 1}{4} \cdot \frac{d^t(s)}{\beta} \cdot X^t(s)  \right) - \sum_{v \in \Gamma^{t}(s)} x_v^* \cdot d^t(s) \,\middle|\, x^{t-1}, F^{t-1}, \Psi^t, \Lambda^t} \label{eq:e^a approx} \text{,}
\end{equation}
where we use the approximation \(e^a \leq 1 + (e^4 - 1) \cdot a / 4\) for \(a \in [0,4]\), which applies since \(v \in \Gamma^{t}(s)\) implies that there exists \(i\) such that \(c_v \geq \kappa_{s_i}^t / 2 \geq d^t(s) / 4\) by definition.
Since \(\sum c_v \cdot x_v^{t-1} = \beta\) and \(\log 1+ y \leq y\), we bound \ref{eq:e^a approx} by 
\begin{align}
    &\leq \expectover{s \sim U^{t-1}}*{ \frac{e^4 - 1}{4} \cdot d^t(s) \cdot X^t(s) - \sum_{v \in \Gamma^{t}(s)} x_v^* \cdot d^t(s) \,\middle|\, x^{t-1}, F^{t-1}, \Psi^t, \Lambda^t} \nonumber \\
    &\leq \expectover{s \sim U^{t-1}}*{ \frac{e^4 - 1}{4} \cdot d^t(s) \cdot \min(X^t(s), 1) - \sum_{v \in \Gamma^{t}(s)} x_v^* \cdot d^t(s) \,\middle|\, x^{t-1}, F^{t-1}, \Psi^t, \Lambda^t} \label{eq:Lambda holds} \text{,}
\end{align}
where step~\eqref{eq:Lambda holds} follows by the definition of \(\Lambda^t\).

Combining \cref{eq:Lambda holds} with the remaining parts of \(\Phi_L\) gives 
\begin{align}
    &\expectover{s^t, R^t}*{\Phi_L(t) - \Phi_L(t-1) \,\middle|\, x^{t-1}, F^{t-1}, U^{t-1}, \Psi^t, \Lambda^t} \nonumber \\
    &\leq \expectover{s \sim U^{t-1}}*{ \frac{e^4 - 1}{4} \cdot d^t(s) \cdot \min(X^t(s), 1) - \sum_{v \in \Gamma^{t}(s)} x_v^* \cdot d^t(s) - 8 \sum_{v \in D^t} c_v x_v^* + 2^{\ell-2} \,\middle|\, x^{t-1}, F^{t-1}, \Psi^t, \Lambda^t} \label{eq:learning bound before casework} \text{.}
\end{align}

We now want to bound this by applying \cref{lemma:alternate}.
To use this lemma, we need to know that $D^t$ and $\Gamma^t(s)$ form a cover of $B_{G/F^{t-1}}(s_0, 2^{\ell-3}) \eqqcolon B(s_0, 2^{\ell-3})$. 
First, note that $D^t$ contains all nodes in $B(s_0, 2^{\ell-3})$ which have not yet been claimed. This means the only nodes left to worry about are nodes in $B(s_0, 2^{\ell-3}) \setminus D^t$. 
Let $u \in B(s_0, 2^{\ell-3}) \setminus D^t$. 
Since $u$ is not in $D^t$, it was claimed on level $\ell$ in a previous round. 
Therefore, it must be distance at most $2^{\ell-3}$ from some previous terminal $s_0'$. 
If we bought node $u$, then $d_{G / (F^{t-1} \cup \{u\})}(s_0,s_0') \leq 2^{\ell-2}$. By assumption $\Psi^t$, we know $d_{G / F^{t-1}}(s_0,s_0') \geq d^t(s)/2 \geq 2^{\ell-1}$. 
Hence $u$ would reduce the cost of $\kappa^t_{s_0}$ by at least a factor of 2, meaning $u \in \Gamma^t(s)$. 
Thus, we can use \cref{lemma:alternate} on the cover $\{ D^t, \Gamma^t(s) \}$.

We then use the second part of \cref{lemma:alternate} to change the cost of all nodes in $\Gamma^t(s)$ to be $\frac{1}{8}d^t(s) \geq 2^{\ell-3}$. 
From this lemma, we get that either
\begin{equation*}
    8 \cdot\sum_{v \in D^t} c_v \cdot x_v^* \geq 8 \cdot 2^{\ell-3} /2 = 2^{\ell-1}
\end{equation*}
or
\begin{equation*}
    8 \cdot \sum_{v \in \Gamma^t(s)} \frac{1}{8} d^t(s) \cdot x_v^* \geq 8 \cdot 2^{\ell-3} /2 = 2^{\ell-1} \text{.}
\end{equation*}
Since we know $2^{\ell-1} \geq d^t(s)/4$, and both terms are non-positive, we can use this combined with the fact that \(2^{\ell} \leq d^t(s)\) to bound \eqref{eq:learning bound before casework} as
\begin{align}
    &\expectover{s^t, R^t}*{\Phi_L(t) - \Phi_L(t-1) \,\middle|\, x^{t-1}, F^{t-1}, U^{t-1}, \Psi^t, \Lambda^t} \nonumber \\
    &\leq \expectover{s \sim U^{t-1}}*{ \frac{e^4 - 1}{4} \cdot d^t(s) \cdot \min(X^t(s), 1) - \frac{d^t(s)}{8} \,\middle|\, F^{t-1}, \Psi^t, \Lambda^t} \label{eq:final line learn} \text{,}
\end{align}
which concludes the proof.
\end{proof}

\changecover*

\begin{proof}
We start by considering the expected change to \(\Phi_C\) over the randomness of the sampling.
Fix an arriving client \(\sigma\).
Expanding definitions,
\begin{align}
    &\expectover{R^t}*{\Phi_C(t) - \Phi_C(t-1) \,\middle|\, x^{t-1}, F^{t-1}, U^{t-1}, s = \sigma} \nonumber \\
    &= \beta \cdot \expectover{R^t}*{\log \left( \frac{\frac{\rho^t}{\beta} + \frac{1}{n}}{\frac{\rho^{t-1}}{\beta} + \frac{1}{n}} \right) \,\middle|\, x^{t-1}, F^{t-1}, U^{t-1}, s = \sigma} \nonumber \\
    &= \beta \cdot \expectover{R^t}*{\log \left( \frac{\rho^t + \frac{\beta}{n}}{\rho^{t-1} + \frac{\beta}{n}} \right) \,\middle|\, x^{t-1}, F^{t-1}, U^{t-1}, s = \sigma} \nonumber \\
    &= \beta \cdot \expectover{R^t}*{\log \left(1 -  \frac{\rho^{t-1} - \rho^t}{\rho^{t-1} + \frac{\beta}{n}} \right) \,\middle|\, x^{t-1}, F^{t-1}, U^{t-1}, s = \sigma} \nonumber \\
    &\leq -\beta \cdot \frac{1}{\rho^{t-1} + \frac{\beta}{n}} \cdot \expectover{R^t}*{\rho^{t-1} - \rho^t \,\middle|\, x^{t-1}, F^{t-1}, U^{t-1}, s = \sigma} \label{eq:cover log approx} \text{.}
\end{align}
Above, \ref{eq:cover log approx} follows from the approximation \(\log(1-y) \leq -y\).
Expanding the definition of \(\rho^t\), \ref{eq:cover log approx} is bounded by
\begin{align}
    &\leq -\beta \cdot \frac{1}{\rho^{t-1} + \frac{\beta}{n}} \cdot \expectover{R^t}*{\sum_{u \in U^{t-1}} (\kappa_{u}^{t} - \kappa_{u}^{t+1}) \,\middle|\, x^{t-1}, F^{t-1}, U^{t-1}, s = \sigma} \label{eq:rho monotone dec} \\
    &\leq -\beta \cdot \frac{1}{\rho^{t-1} + \frac{\beta}{n}} \cdot \expectover{R^t}*{\sum_{u \in U^{t-1}} \frac{\kappa_{u}^{t}}{38} \cdot \mathds{1}\{\Gamma^{t}(u) \cap R^t \neq \emptyset\} \,\middle|\, x^{t-1}, F^{t-1}, U^{t-1}, s = \sigma} \label{eq:kappa ratio} \\
    &=  -\beta \cdot \frac{1}{\rho^{t-1} + \frac{\beta}{n}} \cdot \sum_{u \in U^{t-1}} \frac{\kappa_u^{t}}{38} \cdot  \probover{R^t}*{ \mathds{1}\{\Gamma^{t}(u) \cap R^t \neq \emptyset\} \,\middle|\, x^{t-1}, F^{t-1}, U^{t-1}, s = \sigma} \nonumber \\
    &\leq  -\beta \cdot \frac{1}{\rho^{t-1} + \frac{\beta}{n}} \cdot \sum_{u \in U^{t-1}} \frac{\kappa_u^{t}}{38} \cdot (1 - e^{-1}) \cdot \min \left( \frac{d^t(s)}{\beta} \cdot \sum_{v \in \Gamma^{t}(u)} x_v^{t-1}, 1   \right) \label{eq:convexity inequality} \\
    &\leq - \frac{1 - e^{-1}}{38} \cdot d^t(s) \cdot \frac{\abs{U^{t-1}}}{\rho^{t-1} + \frac{\beta}{n}} \cdot \expectover{u \sim U^{t-1}}*{\kappa_u^{t} \cdot \min(X^t(u), 1) \,\middle|\, x^{t-1}, F^{t-1}} \label{eq:rewrite expectation} \\
    &\leq - \frac{1 - e^{-1}}{456} \cdot \kappa_s^t \cdot \frac{\abs{U^{t-1}}}{\rho^{t-1} + \frac{\beta}{n}} \cdot \expectover{u \sim U^{t-1}}*{d^{t}(u) \cdot \min(X^t(u), 1) \,\middle|\, x^{t-1}, F^{t-1}} \label{eq:kappa vs d} \text{.}
\end{align} 
step~\eqref{eq:rho monotone dec} follows since the terms of \(\rho\) are monotone decreasing when fixing a terminal.
step~\eqref{eq:kappa ratio} follows since \(2^{\ell+3} \leq \kappa_u^t < 3 \cdot 2^{\ell+2}\), so when we sample something in \(\Gamma\), we either stay on the same level and decrease some \(\kappa_{u_i}^t\) satisfying \(2^{\ell - 1} \leq d^t(u)/2 \leq \kappa_{u_i}^t \leq d^t(u)\) by at least factor of \(1/2\), dropping \(\kappa_u\) by a factor of at least \(37/38\), or we decrease \(\ell\), dropping \(\kappa_u\) by a factor of at least \(3/4\).
step~\eqref{eq:convexity inequality} is due to the fact that each node \(v\) is sampled independently with probability \(\min(d^t(s) \cdot x_v^{t-1} / \beta, 1)\), so the probability any given client \(u \in U^{t-1}\) gets at least one node from \(\Gamma^{t}(u)\) is 
\begin{align}
    1 - \prod_{v \in \Gamma^{t}(u)} \left(1 - \min\left(\frac{d^t(s) \cdot x_v^{t-1}}{\beta}, 1\right)\right) &\geq 1 - \exp \left( -\min\left(\frac{d^t(s)}{\beta} \sum_{v \in \Gamma^{t}(u) } x_v^{t-1}, 1\right) \right) \nonumber \\
    &\geq (1 - e^{-1}) \cdot \min\left(\frac{d^t(s)}{\beta} \sum_{v \in \Gamma^{t}(u) } x_v^{t-1}, 1\right) \label{eq:exponential convexity} \text{.}
\end{align}
Above, \ref{eq:exponential convexity} follows from convexity of the exponential.
step~\eqref{eq:rewrite expectation} then follows by rewriting the sum as an expectation and using the fact that \(d^t(s) / \beta \leq 1\).
Finally, step~\eqref{eq:kappa vs d} follows from recalling \(2^{\ell} \leq d^t(s) < 2^{\ell+1} < 2^{\ell+3} \leq \kappa_s^t < 3 \cdot 2^{\ell+2}\).

Taking the expectation of \ref{eq:rewrite expectation} over \(s \sim U^{t-1}\), and using the fact that \(\expectover{s \sim U^{t-1}}*{\kappa_s^{t} \,\middle|\, x^{t-1}, F^{t-1}} = \rho^{t-1} / \abs{U^{t-1}}\), the expected change in \(\Phi_C\) becomes 
\begin{align}
    &\expectover{s^t, R^t}*{\Phi_C(t) - \Phi_C(t-1) \,\middle|\, x^{t-1}, F^{t-1}, U^{t-1}} \nonumber \\
    &\leq -\frac{1 - e^{-1}}{456} \cdot \frac{\rho^{t-1}}{\rho^{t-1} + \frac{\beta}{n}} \cdot \expectover{s \sim U^{t-1}}*{d^t(s) \cdot \min(X^t(s), 1) \,\middle|\, x^{t-1}, F^{t-1}} \nonumber \\
    &\leq -\frac{1 - e^{-1}}{912} \cdot \expectover{s \sim U^{t-1}}*{d^t(s) \cdot \min(X^t(s), 1) \,\middle|\, x^{t-1}, F^{t-1}} \label{eq:min con cost} \text{,}
\end{align}
where in \ref{eq:min con cost} we use the fact that any connection which we must purchase must cost at least $\beta/n$, as we already bought all connections which cost less.
Since this is the case, we know that \(\rho^{t-1} \geq \kappa_{s}^{t} \geq d^t(s) \geq \beta/n\), and so \(\frac{\rho^{t-1}}{\rho^{t-1} + \frac{\beta}{n}} \geq \frac{1}{2}\).\end{proof}

\subsection{Proofs from \cref{sec:pcst_forest}}

\pcpotentialbounds*

\begin{proof}
By our work in the proof of \cref{lemma:potential_bounds}, we only need to inspect the two new terms of the learning potential and the slightly modified covering potential.
Since \(\sum_{s \in U^0}\pi_s \cdot z_s^* \leq \beta\) and \(\sum_{i=1}^k \beta/i = O(\beta \cdot \log k)\), the upper bound holds for the learning potential, and the lower bound holds for the learning potential as all its terms are non-negative.

Since we may assume \(\pi_s,d^0_s \leq \beta\) for all clients \(s\),
\begin{equation*}
    \beta \cdot \log_2 (10k) \cdot \log \left( \frac{\rho^t}{\beta} + \frac{1}{k} \right) = O( \beta \cdot \log^2k) \text{,}
\end{equation*}
and hence the upper bound holds.
Since \(\rho^t \geq 0\) for all \(t\), \(\Phi_C(t) \geq \beta \cdot \log_2(10 k) \cdot \log \left( \frac{\rho^t}{\beta} + \frac{1}{k} \right) \geq - \beta \cdot \log_2(10k) \cdot \log k = - O(\beta \cdot \log^2 k)\), and hence the lower bound holds.
\end{proof}

\pcchangelearn*

\begin{proof}
If \(\neg \Xi^t\) holds, then since \(x^t = x^{t-1}\), inequality~\eqref{eq:pc learn bound 2} follows trivially.
Hence we assume \(\Xi^t\) holds.
By following the proof of \cref{lemma:change in learn} until \eqref{eq:learning bound before casework}, we have 
\begin{align}
    &\expectover{s^t, R^t}*{\Phi_L(t) - \Phi_L(t-1) \,\middle|\, x^{t-1}, F^{t-1}, U^{t-1}, \Psi^t, \Lambda^t, \Xi^t, Y^t = 0} \nonumber \\
    &\leq \expectover{s \sim U^{t-1}}*{ 
    \begin{array}{ll}
    \displaystyle \frac{e^4 - 1}{4} \cdot d^t(s) \cdot \min(X^t(s), 1) \\
    - \displaystyle \sum_{v \in \Gamma^{t}(s)} x_v^* \cdot d^t(s) 
    -  8 \cdot \sum_{v \in D^t} c_v \cdot x_v^* + 2^{\ell-2} 
    \end{array}
    \,\middle|\, x^{t-1}, F^{t-1}, \Psi^t, \Lambda^t, \Xi^t, Y^t = 0} \label{eq:pc learning bound before casework} \text{.}
\end{align}
As shown in the proof of \cref{lemma:change in learn}, $\{ D^t, \Gamma^t(s) \}$ forms a cover of $B_{G/F^{t-1}}(s_0, 2^{\ell-3}) \eqqcolon B(s_0, 2^{\ell-3})$.  
Thus, we can use \cref{lemma:pc-alternate}.

We use the second part of \cref{lemma:pc-alternate} to change the cost of all nodes in $\Gamma^t(s)$ to be $\frac{1}{8}d^t(s) \geq 2^{\ell-3}$. 
From this lemma, we get that either
\begin{equation*}
    8 \cdot\sum_{v \in D^t} c_v \cdot x_v^* \geq (8 \cdot 2^{\ell-3} /2) \cdot (1-z^*_s) = 2^{\ell-1}\cdot (1-z^*_s)
\end{equation*}
or
\begin{equation*}
    8 \cdot \sum_{v \in \Gamma^t(s)} \frac{1}{8} d^t(s) \cdot x_v^* \geq (8 \cdot 2^{\ell-3} /2) \cdot (1-z^*_s) = 2^{\ell-1}\cdot (1-z^*_s) \text{.}
\end{equation*}
Since we know $2^{\ell-1} \geq d^t(s)/4$, and both terms are non-positive, we can use this combined with the fact that \(2^{\ell} \leq d^t(s)\) to bound our above inequality as follows:
\begin{align}
    &\expectover{s^t, R^t}*{\Phi_L(t) - \Phi_L(t-1) \,\middle|\, x^{t-1}, F^{t-1}, U^{t-1}, \Psi^t, \Lambda^t, \Xi^t, Y^t = 0} \nonumber \\
    &\leq \expectover{s \sim U^{t-1}}*{ \frac{e^4 - 1}{4} \cdot d^t(s) \cdot \min(X^t(s), 1) - \frac{d^t(s)}{8}\cdot (1-z^*_s) \,\middle|\, F^{t-1}, \Psi^t, \Lambda^t, \Xi^t, Y^t = 0} \label{eq:pc learn final line buy path} \text{.}
\end{align}
Moving to the case where \(Y^t = 1\), since \(x^t = x^{t-1}\), we have 
\begin{align}
    &\expectover{s^t, R^t}*{\Phi_L(t) - \Phi_L(t-1) \,\middle|\, x^{t-1}, F^{t-1}, U^{t-1}, \Psi^t, \Lambda^t, \Xi^t, Y^t = 1} \nonumber \\ 
    &\leq \expectover{s^t, R^t}*{-\pi_s \cdot z_s^* \,\middle|\, x^{t-1}, F^{t-1}, U^{t-1}, \Psi^t, \Lambda^t, \Xi^t, Y^t = 1} \nonumber \\
    &\leq \expectover{s \sim U^{t-1}}*{ \frac{e^4 - 1}{4} \cdot \pi_s \cdot \min(X^t(s), 1) - \frac{\pi_s}{8}\cdot z^*_s \,\middle|\, F^{t-1}, \Psi^t, \Lambda^t, \Xi^t, Y^t = 1} \label{eq:pc learn final line pay penalty} \text{.}
\end{align}
Since \(\Pr[Y^t = 1] = \frac{d^t(s)}{\pi_s + d^t(s)}\), inequality~\eqref{eq:pc learn bound 1} follows by combining lines \ref{eq:pc learn final line buy path} and \ref{eq:pc learn final line pay penalty}.
\end{proof}

\pcchangecover*

\begin{proof}
We start by considering the expected change to \(\Phi_C\) over the randomness of the sampling.
Fix an arriving client \(\sigma\).
Expanding definitions and following the proof of \cref{lemma:change in cover} until \cref{eq:rho monotone dec},
\begin{align}
    &\frac{1}{\log_2 (10k)} \cdot \expectover{R^t}*{\Phi_C(t) - \Phi_C(t-1) \,\middle|\, x^{t-1}, F^{t-1}, U^{t-1}, \Xi^t, s = \sigma} \nonumber \\
    &\leq -\beta \cdot \frac{1}{\rho^{t-1} + \frac{\beta}{k}} \cdot \expectover{R^t}*{\sum_{u \in U^{t-1}} (\kappa_{u}^{t} - \kappa_{u}^{t+1}) \,\middle|\, x^{t-1}, F^{t-1}, U^{t-1}, \Xi^t, s = \sigma} \nonumber \\
    &\leq -\beta \cdot \frac{1}{\rho^{t-1} + \frac{\beta}{k}} \cdot \expectover{R^t}*{\sum_{u \in U^{t-1}} \frac{\min(\pi_u, d^t(u))}{38} \cdot \mathds{1}\{\Gamma^{t}(u) \cap R^t \neq \emptyset\} \,\middle|\, x^{t-1}, F^{t-1}, U^{t-1}, \Xi^t, s = \sigma} \label{eq:kappa ratio pc} \\
    &=  -\beta \cdot \frac{1}{\rho^{t-1} + \frac{\beta}{k}} \cdot \sum_{u \in U^{t-1}} \frac{\min(\pi_u, d^t(u))}{38} \cdot  \probover{R^t}*{ \mathds{1}\{\Gamma^{t}(u) \cap R^t \neq \emptyset\} \,\middle|\, x^{t-1}, F^{t-1}, U^{t-1}, \Xi^t, s = \sigma} \nonumber \\
    &\leq - \frac{1 - e^{-1}}{38} \cdot \frac{d^{t}(s) \cdot \pi_s}{d^t(s) + \pi_s} \cdot \frac{\abs{\{u \in U^{t-1} : \Xi^t(u) \}}}{\rho^{t-1} + \frac{\beta}{k}} \cdot \expectover{u \sim U^{t-1}}*{\min(\pi_u, d^t(u)) \cdot \min(X^t(u), 1) \,\middle|\, x^{t-1}, F^{t-1}, \Xi^t} \label{eq:pc sampling} \\
    &\leq - \frac{1 - e^{-1}}{38} \cdot \frac{d^t(s) \cdot \pi_s}{d^t(s) + \pi_s} \cdot \frac{\abs{\{u \in U^{t-1} : \Xi^t(u) \}}}{\rho^{t-1} + \frac{\beta}{k}} \cdot \expectover{u \sim U^{t-1}}*{\frac{d^{t}(u) \cdot \pi_u}{d^t(u) + \pi_u} \cdot \min(X^t(u), 1) \,\middle|\, x^{t-1}, F^{t-1}, \Xi^t} \label{eq:pc pi vs prob} \text{.}
\end{align} 
step~\eqref{eq:kappa ratio} follows since our work in the proof of \cref{lemma:change in cover} tells us
\begin{equation*}
    \frac{\kappa^{t}_{s_0} + \kappa^{t}_{s_1} + 2^{\ell^t+3}}{\kappa^{t-1}_{s_0} + \kappa^{t-1}_{s_1} + 2^{\ell^{t-1}+3}} \leq 37/38 \text{.}
\end{equation*}
Note that it is readily checked that if \(d^{t-1}(s) > \pi_s\) and \(d^t(s) \leq \pi_s\), then \(\kappa_s\) decreases by a factor of at least 37/38.
step~\eqref{eq:pc sampling} follows by observing that since we assumed \(\Xi^t\), the nodes \(v\) are sampled independently with probability \(\min \left( \frac{d^{t}(s) \cdot \pi_s}{d^t(s) + \pi_s} \cdot x_v^{t-1}  / \beta, 1  \right)\), so our work in the proof of \cref{lemma:change in cover} gives us this bound.
Finally, step~\eqref{eq:pc pi vs prob} follows since \(\min(\pi_s, d^t(s)) \geq \frac{d^t(s) \cdot \pi_s}{d^t(s) + \pi_s}\).

Taking the expectation of \ref{eq:pc pi vs prob} over \(s \sim U^{t-1}\) conditioned on the event \(\Xi^t\), inequality~\eqref{eq:change in cover not preconnected pc} follows.
\end{proof}
    \section{Adversarial Order}
\label{subsec:adv-order-online-rounding}

In this section, we show how to recover a relatively simple $O(\log n \log k)$-competitive algorithm for online \nwst using the ideas from this paper. The idea is simply to replace the LP-solver from \cref{sec:offline} with the online primal-dual algorithm of \cite{DBLP:journals/mor/BuchbinderN09}.

\begin{theorem}[\cite{DBLP:journals/mor/BuchbinderN09}]
    There is an algorithm that, for any covering linear program for which the constraints $\langle a_1, x \rangle \geq 1, \ldots \langle a_n, x \rangle \geq 1$ are revealed online, maintains a monotonically increasing sequence of solutions $x^1 \leq \ldots \leq x^n$ such that each $x^t$ is feasible for the constraints up to time $t$, and $x^n$ is $O(\log m)$-competitive.
\end{theorem}

In order to use this result, we first need to describe an equivalent pure-covering formulation of our flow based LP from \cref{sec:prelims}. Define
$\mathcal C_s := \{S\subseteq V : s\in S,\ r\notin S\}$. For $S\subseteq V$, let $\delta(S) := \{v\in V\setminus S : \exists u\in S \text{ with } uv\in E\}$. Then the cut-covering LP for \nwst is:
$$
    \min \sum_{v\in V} c_v x_v
    \qquad
    \text{s.t.}\qquad
    \sum_{v\in\delta(S)}x_v\geq 1
    \quad \forall s\in T,\ \forall S\in\mathcal C_s,
    \qquad
    x_v\geq 0.
$$

\begin{lemma}
\label{lem:path-cut-equivalence}
For every vector $x\in \mathbb R^V_{\geq 0}$ and terminal $s$, the following are
equivalent:
\begin{enumerate}
    \item The flow LP: there is one unit of $r$--$s$ flow whose vertex congestion at every
    vertex $v$ is at most $x_v$;
    \item The cut LP: for every $S\subseteq V$ with $s\in S$ and $r\notin S$,
    $$\sum_{v\in \delta(S)} x_v\geq 1.$$
\end{enumerate}
Consequently, the path relaxation and the cut relaxation have the same optimum.
\end{lemma}

\begin{proof}

It suffices to show that every solution given by the flow LP satisfies the cut constraints; and for every cut LP solution, there exists a feasible flow LP solution with the same objective value.  

Let $(x, f)$ be the solution to the flow LP, we show that $x \in \mathbb R^V_{\geq 0}$ satisfies the cut constraints. Fix an arbitrary terminal $s$ and an arbitrary cut $S\in \mathcal C_s$, we have 
\[
    1 \leq \sum_{P \in \mathcal P_s} f_P \leq \sum_{v \in \delta(S)} \sum_{P\in \mathcal P_s : v\in P} f_P \leq \sum_{v\in \delta(S)} x_v,
\]
where the inequalities are by the flow LP constraints; and the equality is because every flow has to enter the set $S$ through some vertex in $\delta(S)$. Hence, the cut constraints are satisfied. 

Let $x$ be a cut LP solution, we show that there exists a flow LP solution $(x, f)$ with the same objective value. We can view $x$ as vertex capacities. By max-flow/min-cut, there is one unit of $r$--$s$ flow respecting vertex capacities $x$ if and only if every $s$--$r$ cut has capacity at least $1$, which is exactly the cut constraints. Hence, there exists a flow LP solution $(x, f)$ with the same objective value.
\end{proof}

\paragraph{Online Rounding Algorithm}

It remains to show how to round the sequence of online fractional solutions to integral solutions, and for this we use a standard threshold sampling strategy. For every vertex $v$, independently sample $\theta_v \sim \mathrm{Unif}[0,1]$. Let $C$ be a sufficiently large constant and define $\rho_t := C\ln(t+1)$, note that this quantity is also monotone in $t$.  After round $t$, define the threshold-rounded set 
$$R^t := \{v\in V : \theta_v \leq \min(1,\rho_t \cdot x^t_v)\}.
$$ 
Since both $\rho_t$ and $x^t$ are monotone, the sets $R^t$ are monotone. When terminal $s^t$ arrives, the algorithm first updates the online covering solution, adds the newly rounded vertices to $R^t$, and then buys a shortest path $P^t$ from $s^t$ to the root in $G/(F^{t-1}\cup R^t)$.

\begin{algorithm}
\caption{\textsc{ThresholdOnlineNWST}}
\label{alg:online-adv}
\begin{algorithmic}[1]
\State Sample $\theta_v \sim \mathrm{Unif}[0,1]$ independently for every $v\in V$
\State $F^0 \gets \{r\}$ and $R^0 \gets \emptyset$
\For{terminal $s^t$ at time $t$}
    \State Update the online covering solution $x^t$ using \cite{DBLP:journals/mor/BuchbinderN09}.
    \State Update $R^t \gets R^{t-1}\cup \{v: \theta_v \leq \min(1,\rho_t \cdot x^t_v)\}$.
    \State Let $P^t$ be the shortest $r$--$s^t$ path in $G/(F^{t-1}\cup R^t)$.
    \State $F^t \gets F^{t-1}\cup R^t\cup P^t$.
\EndFor
\State \Return $F^k$
\end{algorithmic}
\end{algorithm}

\begin{theorem}
\label{thm:adv-online-nwst}
\cref{alg:online-adv} is $O(\log n\log k)$-competitive for adversarial-order online \nwst against an oblivious adversary.
\end{theorem}

The cost of the threshold rounding is at most $O(\rho_k) \cdot c^\intercal x^k = O(\log n \log k \cdot \lpopt)$, so to prove the theorem, we need only bound the cost of the greedy paths. For this, we use exactly the same shortcut sets as in the offline analysis. Let $s^0:=r$. For every $t$ and integer $\ell$, define
$$\Gamma^t(\ell) := \bigcup_{0\leq i<t} B(s^i,2^{\ell-2})\cap B(s^t,2^{\ell-2}).$$
These sets depend only on the graph and the terminal sequence and are fixed before the thresholds are sampled.

Let $\ell^t := \lfloor \log_2 c(P^t)\rfloor$. Call round $t$ cheap if $c(P^t)\leq \lpopt/t$. The total cost of cheap rounds is at most $\sum_{t=1}^k  \lpopt / t =O(\lpopt \cdot \log k),$ which is negligible compared to the final guarantee, so we can assume every round is non-cheap. For every non-cheap round, $c(P^t)\in (\lpopt/t,\lpopt]$, so throughout the $k$ rounds there are only $O(\log k)$ possible relevant cost levels, just like in the offline section. 
Let $\mathcal E_t$ be the event that round $t$ has $x^t(\Gamma^t(\ell^t))\geq 1/2$. We now bound the expected cost of all the greedy paths $P^t$.

\begin{lemma}[Greedy paths cost]
\label{lem:online-greedy-path-cost}
\[
\expect*{\sum_{t=1}^k c(P^t)} \leq O(\log k)\cdot c^\intercal x^k \leq O(\log k \log n) \cdot \lpopt.
\]
\end{lemma}
\begin{proof}
    We do case analysis, per round, based on whether $\mathcal E_t$ occurs or not. Fix an arbitrary round $t$. 
  
    We first show that $R^t\cap \Gamma^t(\ell^*)=\emptyset$ happens with a small probability.
        Fix a round $t$. If there is no level $\ell$ such that $x^t(\Gamma^t(\ell))\geq 1/2$, then the claim is immediate. Otherwise, let $\ell^*$ be the minimum such level. Since $\Gamma^t(\ell^*)$ and $x^t$ are fixed independently of the thresholds,
        $$
        \Pr[R^t\cap \Gamma^t(\ell^*)=\emptyset] 
        = \prod_{v\in \Gamma^t(\ell^*)} \left( 1 - \min(1,\rho_t \cdot x^t_v) \right)
        \leq \exp(- \rho_t / 2),
        $$
        where the inequality is by the fact that $1-z \leq e^{-z}$ for all $z\in \mathbb R$. 
        Note that if some vertex $v \in \Gamma^t(\ell^*)$ satisfies $\rho_t \cdot x^t_v\geq 1$, then the miss probability is zero; otherwise the exponent is at least $\rho_t \cdot x^t(\Gamma^t(\ell^*))\geq \rho_t/2$.

        We now show that $\mathcal E_t$ can only occur if $R^t\cap \Gamma^t(\ell^*)=\emptyset$.
        Suppose $R^t\cap \Gamma^t(\ell^*)\neq \emptyset$, and fix $v\in R^t\cap \Gamma^t(\ell^*)$. Then for some $i<t$, $v\in B(s^i,2^{\ell^*-2})$. 
        Hence, $c(P^t) \leq d(s^t, v) + d(v, s^i) \leq 2^{\ell^* - 2} + 2^{\ell^* - 2} = 2^{\ell^* - 1}.$
        Thus $\ell^t<\ell^*$. By minimality of $\ell^*$, this implies $x^t(\Gamma^t(\ell^t))<1/2$. Therefore the event $\mathcal E_t$ can only occur if $R^t\cap \Gamma^t(\ell^*)=\emptyset$, which happens with probability at most $\exp(-\rho_t/2)$.

        The cost of rounds where $\neg \mathcal E_t$ occurs can be bounded using the same arguments as in the offline case. Since $x^t(\Gamma^t(\ell^t))<1/2$, the total amount of flow routed through $\Gamma^t(\ell^t)$ is less than $1/2$: 
        \[
            \sum_{P : P \cap \Gamma^t(\ell^t) = \emptyset} f_P^t = \sum_{v \in \Gamma^t(\ell^t)} \sum_{P : P \ni v} f_P^t\leq \sum_{v \in \Gamma^t(\ell^t)} x_v^t < 1/2,
        \]
        where the first inequality follows from the LP constraint; and the second inequality is by the definition of $\neg \mathcal{E}_t$.
        
        Define $D^t := B(s^t,2^{\ell^t-2})\setminus \Gamma^t(\ell^t)$. Note that $D^t$ is disjoint from $D^{t'}$ for $t' < t$ and $\ell^{t'} = \ell^t$. Otherwise their intersection vertex would be in both $B(s^t,2^{\ell^t-2})$ and $B(s^{t'},2^{\ell^t-2})$, and hence in $\Gamma^t(\ell^t)$, which is impossible.

        On one hand, the cost of vertices in $D^t$ is, up to a constant factor, an upper bound on $c(P^t)$: 
        \[
            \sum_{v \in D^t} c_v x^t_v \geq \sum_{v \in D^t} c_v \cdot \sum_{P \ni v} f_P^t = \sum_{P}f_P^t \left( \sum_{v \in P \cap D^t}  c(v)\right) \geq \sum_{P}f_P^t \cdot 2^{\ell^t - 3} \geq c(P^t) / 32 ,
        \]
        where the first inequality follows from the LP constraint; the second one follows from the fact that each path in $D^t$ has length at least $2^{\ell^t - 3}$; and the third inequality is by the definition of $\ell^t$.
        
        On the other hand, we know that $\{ D^t \}_{t=1}^k$ is a collection of disjoint sets, distributed over $O(\log k)$ levels. As a result, the total cost of all $D^t$'s is bounded by $c^\intercal x^k$ per level. Therefore, 
        \[
            \sum_{t = 1}^{k}c(P^t) \leq 32 \sum_{t=1}^k \sum_{v \in D^t} c_v x^t_v \leq O(\log k) \sum_{v \in V} c_v x^k_v.
        \]

    Now we can bound the expected cost of all the greedy paths: 
    \begin{align*}
        \expect*{\sum_{t = 1}^k c(P^t)} 
        & = \sum_t \expect*{c(P^t) | \cE_t} \cdot \prob{\cE_t} + \expect*{c(P^t) | \neg\cE_t} \cdot \prob{\neg \cE_t} \\ 
        & \leq \sum_t c^\intercal x^k \cdot \exp(-\rho_t/2) + O(\log k) \cdot c^\intercal x^k \\
        & \leq O(\lpopt) + O(\log k) \cdot c^\intercal x^k = O(\log k)\cdot c^\intercal x^k,
    \end{align*}
    where the first line is by law of total expectation, the second line follows from the fact that $c(P^t) \leq c^\intercal x^k$ regardless of the event $\mathcal E_t$, and the last line is by the definition of $\rho_t$.
\end{proof}

This concludes the proof of the theorem.

    \section{Runtime}

Now that we are satisfied with the correctness of the algorithm, we will analyze its runtime. 
We begin by analyzing the algorithm where we assume to have access to \(\beta\) that satisfies $\lpopt \leq \beta \leq 2 \cdot \lpopt$.
This algorithm is copied below for ease of reference.
In section~\ref{sec:guess-and-double}, we double check that guess-and-double does not impact the asymptotic runtime nor the competitive ratio.

Let $T_{\textsc{SP}}$ denote the runtime of one shortest path computation on a graph $G$ with $|V| = O(n)$ and $|E| = O(m)$.   
Using Fibonacci heaps \cite{DBLP:journals/jacm/FredmanT87} one can implement shortest path in $T_{\textsc{SP}} = O(m + n \log n)$ time. We claim that \cref{alg:nwsf} runs in time $O(k \cdot T_{\textsc{SP}})$, where $k$ is the number of terminals arriving online. 
To see this, it suffices to show that (1) the precomputation takes $O(m)$ time, and (2) the runtime per round is $O(T_{\textsc{SP}})$.

To implement precomputation, we loop over the vertices, and only keep the ones whose costs are within $[\beta/ n, \beta]$. Then we can initializing their weights. Note that there are only $O(\log n)$ many values of $\ell$. Hence, the initialization of node weights $x_v^0$, the forest $F^0$, and the tree per level $T_\ell^0$ takes at most $O(n)$ time. 

Within each iteration when a terminal pair $s = (s_0, s_1)$ arrives, to find $\kappa^t_{s_0}$ and $\kappa^t_{s_1}$, we perform two shortest path computation, centered at $s_0$ and $s_1$ respectively. At the end, we store $d^t(s_i, u)$ into an array for both terminals. These distances would be used later to test membership of $u \in \Gamma^t(s)$. In total, the initialization and shortest paths finding take at most $O(T_{\textsc{SP}})$. The sampling procedure on line 11 to 13 takes $O(n)$. If testing membership in $\Gamma^t(s)$ takes constant time, the rest of the loop can be done by scanning all vertices, and updating their weights accordingly. This incurs linear overhead, which is still $O(T_{\textsc{SP}})$.

It remains to show that, at the cost of $O(T_{\textsc{SP}})$ time to do a multi-source shortest path, we can make testing the membership of the set $\Gamma^t(s)$ in constant time. Recall that given an arbitrary vertex $u \in V$ and terminal $s_i$, querying $d^t(s_i, u)$ takes constant time. We compute a multi-source shortest path from our current level $T_{\ell}^{t - 1}$ by adding a super source node connecting each vertex in the tree. 

Now querying the distance to $T_\ell^t$ and $s_i$ both take constant time, we iterate over the vertices to check whether it belongs to $\Gamma^t(s)$ by the following procedure. For an arbitrary vertex $u \in V$ and arbitrary terminal $s_i$, if $d(s_i, u) + d(u, s_{1 - i}) \leq d^t(s) / 2$, then $u \in \Gamma^t(s)$ by definition. If not, we then check whether $d(s_i, u) + d(u, T^{t - 1}_\ell) \leq \kappa_{s_i}^t / 2$.  Then $u \in \Gamma^t(s)$ if the answer is positive, and vice versa. 

Hence, we conclude that the runtime of our algorithm is at most $O(k \cdot T_{\textsc{SP}})$.
    
    \section{Guess-and-Double}\label{sec:guess-and-double}

We perform a standard doubling strategy on \(\beta\): starting with \(\beta = d(s_0,s_1)\), we run \cref{alg:nwsf}.
Once the cost of our algorithm exceeds \(4C \cdot\beta \cdot \log n\) at time \(t\), where \(C\) is the constant hidden in the competitive ratio of \cref{thm:main_nwsf}, we replace \(\beta\) by $2\beta$ and immediately reinitialize \cref{alg:nwsf}. Let $\cA_1,\ldots,\cA_{M}$ be the copies of \cref{alg:nwsf} created thus, and let $\cA_i$ use guess $\beta_i \in [2^i \cdot \lpopt, 2^{i+1} \cdot \lpopt)$. Denote the requests sets fed to $\cA_1, \ldots, \cA_M$ as $T_1, \ldots, T_M$ respectively, such that $T_1, \ldots, T_k$ partition the full request set $T$. Let $M^*$ be the index of the copy $\cA_i$ whose guess $\beta_i \in [\lpopt, 2\cdot \lpopt)$, and let $H_i$ be the event that $\cA_i$ is started because $\cA_{i-1}$ exceeded its budget.

We now argue that this strategy incurs at most a constant factor overhead. For any $i= M^* + j$, the expected cost of $c(\cA_i) \leq C \cdot \beta_i \cdot \log n$ even after conditioning on $H_i$ (because $\cA_i$ is initialized with fresh randomness and processes the elements not fed to $\cA_1, \ldots, \cA_{i-1}$ in uniformly random order). Hence, $\prob{H_i \mid H_{i-1}} \leq \nicefrac{1}{4}$, and thus $\prob{H_i} \leq 4^{-j}$, which in turn means that 
\begin{align*}
    \expect{c(\cA_i)} &= \expect{c(\cA_i) \mid H_i} \cdot \prob{H_i} \leq \frac{1}{4^i} \cdot 2^i \cdot 4C \cdot \lpopt \cdot \log n = 2^{2-i} \cdot C \cdot \lpopt \cdot \log n.
    \intertext{Thus the total cost is}
    \expect*{c(\alg)} &= \sum_{i=1}^{M^*} \expect*{c(\cA_i)} + \sum_{i=M^*+1}^\infty \expect*{c(\cA_i)} \leq O(\log n \cdot \lpopt) + \sum_{j=1}^\infty  2^{2-i} \cdot C \cdot \lpopt \cdot \log n \\
    &= O(\log n \cdot \lpopt).
\end{align*}

We also note that this guess-and-double strategy preserves the asymptotic runtime of $O(k \cdot T_{\textsc{SP}})$, because each $\cA_i$ incurs runtime $O(|T_i| \cdot T_{\textsc{SP}})$ and thus in aggregate they take time $\sum_i O(|T_i| \cdot T_{\textsc{SP}}) = O(k \cdot T_{\textsc{SP}})$.

    {
    \footnotesize
    \bibliography{refs}

\newcommand{\etalchar}[1]{$^{#1}$}
\begin{thebibliography}{ALM{\etalchar{+}}98}

\bibitem[AAA{\etalchar{+}}06]{DBLP:journals/talg/AlonAABN06}
Noga Alon, Baruch Awerbuch, Yossi Azar, Niv Buchbinder, and Joseph Naor.
\newblock A general approach to online network optimization problems.
\newblock {\em {ACM} Trans. Algorithms}, 2(4):640--660, 2006.

\bibitem[AAB04]{DBLP:journals/tcs/AwerbuchAB04}
Baruch Awerbuch, Yossi Azar, and Yair Bartal.
\newblock On-line generalized steiner problem.
\newblock {\em Theor. Comput. Sci.}, 324(2-3):313--324, 2004.

\bibitem[AKR95]{DBLP:journals/siamcomp/AgrawalKR95}
Ajit Agrawal, Philip~N. Klein, and R.~Ravi.
\newblock When trees collide: An approximation algorithm for the generalized
  steiner problem on networks.
\newblock {\em {SIAM} J. Comput.}, 24(3):440--456, 1995.

\bibitem[ALM{\etalchar{+}}98]{DBLP:journals/jacm/AroraLMSS98}
Sanjeev Arora, Carsten Lund, Rajeev Motwani, Madhu Sudan, and Mario Szegedy.
\newblock Proof verification and the hardness of approximation problems.
\newblock {\em J. {ACM}}, 45(3):501--555, 1998.

\bibitem[BC97]{DBLP:conf/stoc/BermanC97}
Piotr Berman and Chris Coulston.
\newblock On-line algorithms for steiner tree problems (extended abstract).
\newblock In Frank~Thomson Leighton and Peter~W. Shor, editors, {\em
  Proceedings of the Twenty-Ninth Annual {ACM} Symposium on the Theory of
  Computing, El Paso, Texas, USA, May 4-6, 1997}, pages 344--353. {ACM}, 1997.

\bibitem[BEV25]{DBLP:conf/soda/Borst0V25}
Sander Borst, Marek Eli{\'{a}}s, and Moritz Venzin.
\newblock Stronger adversaries grow cheaper forests: online node-weighted
  steiner problems.
\newblock In Yossi Azar and Debmalya Panigrahi, editors, {\em Proceedings of
  the 2025 Annual {ACM-SIAM} Symposium on Discrete Algorithms, {SODA} 2025, New
  Orleans, LA, USA, January 12-15, 2025}, pages 3842--3864. {SIAM}, 2025.

\bibitem[BGRS13]{DBLP:journals/jacm/ByrkaGRS13}
Jaroslaw Byrka, Fabrizio Grandoni, Thomas Rothvo{\ss}, and Laura Sanit{\`{a}}.
\newblock Steiner tree approximation via iterative randomized rounding.
\newblock {\em J. {ACM}}, 60(1):6:1--6:33, 2013.

\bibitem[BHL13]{DBLP:conf/icalp/BateniHL13}
MohammadHossein Bateni, MohammadTaghi Hajiaghayi, and Vahid Liaghat.
\newblock Improved approximation algorithms for (budgeted) node-weighted
  steiner problems.
\newblock In Fedor~V. Fomin, Rusins Freivalds, Marta~Z. Kwiatkowska, and David
  Peleg, editors, {\em Automata, Languages, and Programming - 40th
  International Colloquium, {ICALP} 2013, Riga, Latvia, July 8-12, 2013,
  Proceedings, Part {I}}, Lecture Notes in Computer Science, pages 81--92.
  Springer, 2013.

\bibitem[BN09]{DBLP:journals/mor/BuchbinderN09}
Niv Buchbinder and Joseph Naor.
\newblock Online primal-dual algorithms for covering and packing.
\newblock {\em Math. Oper. Res.}, 34(2):270--286, 2009.

\bibitem[DS14]{DBLP:conf/stoc/DinurS14}
Irit Dinur and David Steurer.
\newblock Analytical approach to parallel repetition.
\newblock In David~B. Shmoys, editor, {\em Symposium on Theory of Computing,
  {STOC} 2014, New York, NY, USA, May 31 - June 03, 2014}, pages 624--633.
  {ACM}, 2014.

\bibitem[Fei98]{feige1998thresholdforsetcover}
Uriel Feige.
\newblock A threshold of \(\ln n\) for appoximating set cover.
\newblock {\em Journal of the ACM}, 45(4):634--652, 1998.

\bibitem[FT87]{DBLP:journals/jacm/FredmanT87}
Michael~L. Fredman and Robert~Endre Tarjan.
\newblock Fibonacci heaps and their uses in improved network optimization
  algorithms.
\newblock {\em J. {ACM}}, 34(3):596--615, 1987.

\bibitem[GK99]{DBLP:journals/iandc/GuhaK99}
Sudipto Guha and Samir Khuller.
\newblock Improved methods for approximating node weighted steiner trees and
  connected dominating sets.
\newblock {\em Inf. Comput.}, 150(1):57--74, 1999.

\bibitem[GKL21]{gupta2024randomordersetcover}
Anupam Gupta, Gregory Kehne, and Roie Levin.
\newblock Random order online set cover is as easy as offline.
\newblock In {\em 62nd {IEEE} Annual Symposium on Foundations of Computer
  Science, {FOCS} 2021, Denver, CO, USA, February 7-10, 2022}, pages
  1253--1264. {IEEE}, 2021.

\bibitem[GKL24]{gupta2023setcoveringeyeswide}
Anupam Gupta, Gregory Kehne, and Roie Levin.
\newblock Set covering with our eyes wide shut.
\newblock In David~P. Woodruff, editor, {\em Proceedings of the 2024 {ACM-SIAM}
  Symposium on Discrete Algorithms, {SODA} 2024, Alexandria, VA, USA, January
  7-10, 2024}, pages 4530--4553. {SIAM}, 2024.

\bibitem[GW95]{DBLP:journals/siamcomp/GoemansW95}
Michel~X. Goemans and David~P. Williamson.
\newblock A general approximation technique for constrained forest problems.
\newblock {\em {SIAM} J. Comput.}, 24(2):296--317, 1995.

\bibitem[HLP17]{DBLP:journals/siamcomp/HajiaghayiLP17}
MohammadTaghi Hajiaghayi, Vahid Liaghat, and Debmalya Panigrahi.
\newblock Online node-weighted steiner forest and extensions via disk
  paintings.
\newblock {\em {SIAM} J. Comput.}, 46(3):911--935, 2017.

\bibitem[IW91]{DBLP:journals/siamdm/ImaseW91}
Makoto Imase and Bernard~M. Waxman.
\newblock Dynamic steiner tree problem.
\newblock {\em {SIAM} J. Discret. Math.}, 4(3):369--384, 1991.

\bibitem[Kor04]{korman2004randomizationinsetcover}
Simon Korman.
\newblock On the use of randomization in the online set cover problem, 2004.

\bibitem[KR95]{DBLP:journals/jal/KleinR95}
Philip~N. Klein and R.~Ravi.
\newblock A nearly best-possible approximation algorithm for node-weighted
  steiner trees.
\newblock {\em J. Algorithms}, 19(1):104--115, 1995.

\bibitem[KZ97]{DBLP:journals/jco/KarpinskiZ97}
Marek Karpinski and Alexander Zelikovsky.
\newblock New approximation algorithms for the steiner tree problems.
\newblock {\em J. Comb. Optim.}, 1(1):47--65, 1997.

\bibitem[NPS11]{DBLP:conf/focs/NaorPS11}
Joseph Naor, Debmalya Panigrahi, and Mohit Singh.
\newblock Online node-weighted steiner tree and related problems.
\newblock In Rafail Ostrovsky, editor, {\em {IEEE} 52nd Annual Symposium on
  Foundations of Computer Science, {FOCS} 2011, Palm Springs, CA, USA, October
  22-25, 2011}, pages 210--219. {IEEE} Computer Society, 2011.

\bibitem[Pan15]{panigrahi2015online}
Debmalya Panigrahi.
\newblock Lecture 18: Online {Steiner} tree.
\newblock Lecture notes for COMPSCI 590.1: Advanced Topics in Computer Science:
  Graph Algorithms, Duke University, March 2015.
\newblock Accessed: 2026-03-31.

\bibitem[PS00]{DBLP:journals/jal/PromelS00}
Hans~J{\"{u}}rgen Pr{\"{o}}mel and Angelika Steger.
\newblock A new approximation algorithm for the steiner tree problem with
  performance ratio 5/3.
\newblock {\em J. Algorithms}, 36(1):89--101, 2000.

\bibitem[RZ00]{DBLP:conf/soda/RobinsZ00}
Gabriel Robins and Alexander Zelikovsky.
\newblock Improved steiner tree approximation in graphs.
\newblock In David~B. Shmoys, editor, {\em Proceedings of the Eleventh Annual
  {ACM-SIAM} Symposium on Discrete Algorithms, January 9-11, 2000, San
  Francisco, CA, {USA}}, pages 770--779. {ACM/SIAM}, 2000.

\bibitem[Zel93]{DBLP:journals/algorithmica/Zelikovsky93}
Alexander Zelikovsky.
\newblock An 11/6-approximation algorithm for the network steiner problem.
\newblock {\em Algorithmica}, 9(5):463--470, 1993.

\end{thebibliography}
    \bibliographystyle{alpha}
    }
\end{document}